\documentclass[article,a4paper,11pt]{article}
\usepackage[english]{babel}
\usepackage[T1]{fontenc}
\usepackage[bitstream-charter]{mathdesign}

\usepackage{multicol}
\usepackage{amsmath}
\usepackage{amssymb}
\usepackage{amsthm}
\usepackage{graphicx}
\usepackage{xcolor}
\usepackage{authblk}
\usepackage{marvosym}
\usepackage{cite}
\usepackage{kotex}
\usepackage{vmargin}
\setmarginsrb{2cm}{2cm}{2cm}{2cm}{0mm}{0mm}{0mm}{10mm}
\usepackage{gensymb}
\usepackage[colorlinks=true, allcolors=blue, breaklinks=true]{hyperref}
\theoremstyle{definition}
\newtheorem{definition}{Definition}
\newtheorem{theorem}{Theorem}
\newtheorem{corollary}{Corollary}
\newtheorem{lemma}{Lemma}
\usepackage{url}
\usepackage{tikz}
\usetikzlibrary{arrows.meta}
\newcommand{\sech}{\mathrm{sech} \,}

\title{\bfseries On spatial Fourier spectrum\\ of rogue wave breathers}
\author{\normalsize Natanael Karjanto\thanks{\Letter: \texttt{natanael@skku.edu} \href{https://orcid.org/0000-0002-6859-447X}{\includegraphics[scale=0.08]{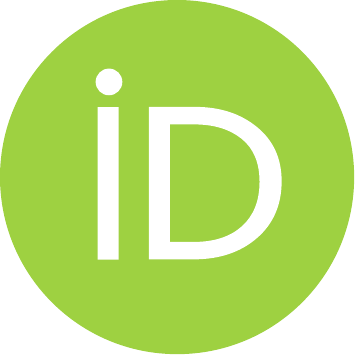}}}}
\affil{Department of Mathematics, University College, Natural Science Campus\\ Sungkyunkwan University, Suwon~16419, Republic of Korea}

\date{\vspace*{-0.5cm} \scriptsize Updated \today}

\begin{document}
\maketitle

\begin{abstract}
\noindent
In this article, we derive exact analytical expressions for the spatial Fourier spectrum of the soliton family on a constant background. Also known as breathers, these solitons are exact solutions of the nonlinear Schr\"odinger equation and are considered as prototypes for rogue wave models. Depending on the periodicity in the spatial-temporal domain, the characteristics in the wavenumber-temporal domain may feature either a continuous or discrete spectrum. \\

\end{abstract}

\section{Introduction}

The study of rogue waves has been of great interest ever since the discovery of soliton solutions of some nonlinear evolution equations. In particular, the focusing nonlinear Schr\"odinger (NLS) equation and its family of exact analytical solutions of solitons on a non-vanishing background, also known as breathers, has been proposed as prototypes for rogue wave modeling. The literature is not exhaustive in disseminating rogue waves study based on the NLS model and its corresponding breather solutions, for which the phenomena are observed and encounter applications in hydrodynamics, optical fibers, and photonic crystals among others~\cite{onorato2013rogue,dudley2014instabilities}. 

Consider the focusing NLS equation in $(1 + 1)$-dimension written in a canonical form~\cite{sulem1999nonlinear,fibich2015nonlinear}:
\begin{equation}
i q_t + q_{xx} + 2|q|^2 q = 0, \qquad (x,t) \in \mathbb{R}^2, \qquad q(x,t) \in \mathbb{C}.		\label{NLSmodel}
\end{equation}
In surface gravity waves, the variables $(x,t)$ represent spatial and temporal quantities, respectively, while in nonlinear optics, $t$ denotes the  transversal pulse propagation in space and $x$ designates the time variable. Unless otherwise indicated, we adopt the interpretation for water waves for the rest of the article. The complex-valued amplitude $q(x,t)$ describes an envelope of the corresponding wave packet profile $\eta$, where usually given by the relationship $\eta(\bar{x},\bar{t}) = \text{Re} \left\{q(x,t) e^{i(k \bar{x} - \omega \bar{t})} \right\}$, where $x = \varepsilon(\bar{x} - c_g \bar{t})$, $t = \varepsilon^2 \bar{t}$, $\varepsilon \ll 1$, $c_g$ is the group velocity, and the wavenumber $k$ and wave frequency $\omega$ are related by the linear dispersion relationship for the corresponding medium.

The simplest nontrivial solution is space-independent, $q(x,t) = q_0(t) = e^{2it}$, known as the continuous or plane wave solution. The simplest solution with dependence on both variables is $q(x,t) = q_0(t) \, \sech(x)$~\cite{dysthe1999note}. In particular, the NLS equation admits the following family of solitons on a constant background, and this is where our interest lies:
\begin{align}
q_K(x,t) &= q_0(t) \left(\frac{\mu^3 \cos(\rho t) + i \mu \rho \sin (\rho t)}{\rho \cosh(\mu x) - 2\mu \cos(\rho t)} - 1 \right), \qquad \rho = \mu \sqrt{4 + \mu^2} \\
q_P(x,t) &= q_0(t) \left(\frac{4 (1 + 4 i t)}{1 + (4t)^2 + (2x)^2} - 1 \right) \\
q_A(x,t) &= q_0(t) \left(\frac{\kappa^3 \cosh(\sigma t) + i \kappa \sigma \sinh (\sigma t)}{2 \kappa \cosh(\sigma t) - \sigma \cos(\kappa x)} - 1 \right), \qquad \sigma = \kappa \sqrt{4 - \kappa^2}.
\end{align}

In our context, $q_K$, $q_P$, and $q_A$ denote the Kuznetsov-Ma (KM), Peregrine/rational, and Akhmediev breathers, respectively~\cite{kuznetsov1977solitons,ma1979perturbed,peregrine1983water,akhmediev1985generation,akhmediev1986modulation,akhmediev1987exact}. The chosen order is more historical rather than its fame, simplicity, or fame, even though $q_P$ is the limiting case for both $q_K$ and $q_A$. See~\cite{karjanto2020peregrine} for an extensive discussion and visual exploration of the limiting behavior of these NLS breathers. While the discussion of these breathers in the spatio-temporal domain has reached maturity, this article fills a gap in covering their features in the wavenumber-temporal domain by observing their spectra. Note that this should not be confused with spectrum in the context of the inverse scattering transform (IST), but rather the physical spectrum that can be calculated by switching from the space domain into the wavenumber domain by means of the spatial Fourier transform.

Recently, \emph{Frontiers of Physics} has published a sequence of topical research articles describing the latest progress of NLS breathers while commemorating the tenth anniversary of the observation of the Peregrine soliton in nonlinear media, including optical fibers, water waves, and plasma physics~\cite{kibler2010peregrine,chabchoub2011rogue,bailung2011observation}. In particular, the stability properties in energy spaces of these three important NLS breathers have been reviewed in~\cite{alejo2020review}. By characterizing the spectral properties of each of these breathers, the authors demonstrated that they are unstable in the Lyapunov sense. 

We begin with the following well-known definitions on Fourier transforms, Fourier series, and spectral decomposition~\cite{pelinovsky2005spectral}.
\begin{definition}[Spatial Fourier transform]		\label{DefSFT}
Let $f(x,t)$ be a square-integrable function on the spatial real line, then it can be represented in a dual spatial-wavenumber $(x,k)$ space by fixing the time variable $t$ as integral transforms
\begin{equation}
f(x,t) = \frac{1}{2\pi} \int_{-\infty}^{\infty} \widehat{f}(k, t) e^{i k x} \, dk,	\label{q5}
\end{equation}
where $\widehat{f}(k,t)$ is the spatial non-unitary Fourier transform written in the terms of angular wavenumber~$k$, and is defined by
\begin{equation}
\widehat{f}(k,t) = \int_{-\infty}^{\infty} q(x,t) e^{-ik x} \, dx.	\label{qhat6}
\end{equation}
\end{definition}

\begin{definition}[Temporal Fourier transform]		\label{DefTFT}
Alternatively, for a square-integrable function $f(x,t)$ on the temporal real line, it can also be expressed at a fixed location $x$ in space as integral transforms between dual temporal-frequency $(t,\omega)$ domain
\begin{equation}
f(x,t) = \frac{1}{2\pi} \int_{-\infty}^{\infty} \widehat{f}(x,\omega) e^{i \omega t} \, d\omega,	\label{q7}
\end{equation}
where $\widehat{f}(x,\omega)$ is the temporal non-unitary Fourier transform written in the terms of angular frequency~$\omega$, and is defined by
\begin{equation}
\widehat{f}(x,\omega) = \int_{-\infty}^{\infty} f(x,t) e^{-i\omega t} \, dt.	\label{qhat8}
\end{equation}
\end{definition}

In this case, the relationship between~\eqref{q5} and~\eqref{qhat6} is between space and wavenumber domains, where the spatial and wavenumber variables become the transform variables for a fixed time. On the other hand, for~\eqref{q7} and~\eqref{qhat8}, they relate between the time and frequency domains at a fixed location in space by letting time and frequency as the transform variables~\cite{bauck2019note}.

\begin{definition}[Spectral decomposition and Complex spatial Fourier series]
For a spatially periodic function $f(x,t)$ with a period $L \neq 0$, i.e., $f(x,t) = f(x + L, t)$, for all $x \in \mathbb{R}$, the spectral decomposition of $f$ is given as a discrete summation of the harmonic oscillation $e^{i k_n x}$
\begin{equation}
f(x,t) = \sum_{-\infty}^{\infty} \widehat{f}_n(t) e^{i k_n x}		\label{f9}
\end{equation}
where $\widehat{f}_n(t) = \widehat{f}(k_n, t)$ are Fourier coefficients of the complex spatial Fourier series for $f(x,t)$.
\begin{equation}
\widehat{f}_n(t) = \frac{1}{2L} \int_{-L}^{L} f(x,t) e^{-i k_n x} \, dx.		\label{fhat10}
\end{equation}
\end{definition}

\begin{definition}[Complex temporal Fourier series]
For a temporally periodic function $f(x,t)$ with a period $T \neq 0$, i.e., $f(x,t) = f(x, t + T)$, $t > 0$, the spectral decomposition of $f$ is given as a discrete summation of the harmonic oscillation $e^{i \omega_n x}$
\begin{equation}
f(x,t) = \sum_{-\infty}^{\infty} \hat{f}_{\omega_n}(x,\omega_n) e^{i \omega_n t}
\end{equation}
where $\hat{f}_{\omega_n}(x) = \hat{f}(x, \omega_n)$ are Fourier coefficients of the complex temporal Fourier series for $f(x,t)$.
\begin{equation}
\hat{f}_{\omega_n}(x) = \frac{1}{2T} \int_{-T}^{T} f(x,t) e^{-i \omega_n t} \, dt.
\end{equation}
\end{definition}

Throughout the rest of this article, unless otherwise mentioned, we focus our discussion on the spatial Fourier transform and (complex) Fourier series for the family of breather soliton solutions of the NLS equation~\eqref{NLSmodel}, i.e.,~\eqref{q5},~\eqref{qhat6},~\eqref{f9}, and~\eqref{fhat10}. Before moving on to the next section, we end this introduction by stating the spatial Fourier transform of a unitary constant function as the following lemma.
\begin{lemma}		\label{LemDiracDelta}
The Fourier transform of the constant function $f(x) = 1$ is given by the Dirac delta function $\delta(k)$, generally in the sense of distributions rather than measures, i.e.,
\begin{equation*}
\widehat{f}(k) = \int_{-\infty}^{\infty} 1 \, e^{-i k x} \, dx = 2\pi \delta(k).
\end{equation*} 
\end{lemma}

The article is organized as follows. Section~\ref{spafoutra} covers the spatial Fourier transform for the breather solitons. It covers the three breathers mentioned earlier and provides analytical expressions for the corresponding spectra. We also illustrate the spectra by plotting them.  Section~\ref{conclude} concludes our discussion.

\section{Spatial Fourier spectrum}		\label{spafoutra}

In this section, we derive the spatial Fourier spectrum for the breather solitons of the NLS equation. For the KM breather and Peregrine soliton, the spectra are continuous in wavenumber~$k$ while for the Akhmediev breather, the spectrum is discrete due to its spatial periodicity. 

\subsection{Kuznetsov-Ma breather}
Before proving the theorem for the spectrum of the KM breather, we need to verify the following lemma.
\begin{lemma}		\label{LemGradRyz}
For $k \neq 0$, $\mu > 0$, and $b > |c| > 0$, we have the following definite integral (Formula {\bfseries 3.983.1(b)} in~\cite{gradshteyn2015table}):
\begin{equation*}
\int_{-\infty}^{\infty} \frac{\cos (k x) \, dx}{b \cosh (\mu x) + c} = 2 \frac{\pi}{\mu} \frac{\sinh \left(\frac{k}{\mu} \arccos \frac{c}{b}\right)}{\sqrt{b^2 - c^2} \sinh \frac{k}{\mu} \pi}.
\end{equation*}
\end{lemma}
\begin{proof}
We consider the complex-valued function 
\begin{equation*}
f(z) = \frac{e^{i k z}}{b \cosh(\mu z) + c}
\end{equation*}
and the rectangular contour $\gamma_{\text{\tiny $K$}}$ with corners at $\pm R$ and $\pm R + 2\pi i/\mu$, for $R > 0$. See Figure~\ref{rectcont}. To find the poles and residues of $f$, we need to solve $b \cosh (\mu z) + c = 0$, $z \in \mathbb{C}$. This gives
\begin{align*}
\frac{b}{2} \left(e^{2 \mu z} + 1 \right) + c e^{\mu z} &= 0 \\
e^{\mu z} &= \frac{-c \pm \sqrt{c^2 - b^2}}{b} = -\frac{c}{b} \pm i \sqrt{1 - \frac{c^2}{b^2}} = \cos(\pi \mp \theta) \pm i \sin(\pi \mp \theta) \\
z &= \ln 1 + \frac{i}{\mu} \left(\pi \mp \theta + 2 n \pi \right), \qquad n \in \mathbb{Z}
\end{align*}
where for $0 \leq \theta < \pi/2$, we take
\begin{equation*}
\cos \theta = \frac{c}{b}, \qquad \sin \theta = \sqrt{1 - \frac{c^2}{b^2}}, \qquad \tan \theta = \sqrt{\frac{b^2}{c^2} - 1} > 0.
\end{equation*}

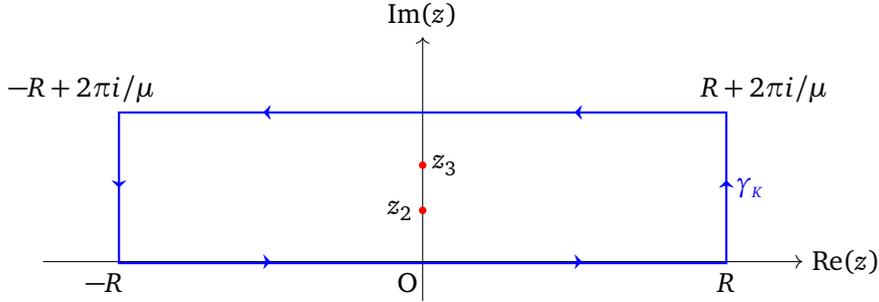
\begin{figure}[htbp]
\begin{center}
\begin{tikzpicture}
\draw[->] (-5,-0.98)--(5,-0.98) node[right]{Re($z$)};
\draw[->] (0,-1.5)--(0,2) node[above]{Im($z$)};
\node[draw,thick,rectangle,minimum width=8cm,minimum height=2cm,color=blue](r){};
\draw[-{Stealth[length=1mm,width=2mm]},color=blue] (4,0)--(4,0.1);
\draw[-{Stealth[length=1mm,width=2mm]},color=blue] (-4,0.1)--(-4,0);
\draw[-{Stealth[length=1mm,width=2mm]},color=blue] (2.1,1)--(2,1);
\draw[-{Stealth[length=1mm,width=2mm]},color=blue] (-2,1)--(-2.1,1);
\draw[-{Stealth[length=1mm,width=2mm]},color=blue] (-2.1,-0.98)--(-2,-0.98);
\draw[-{Stealth[length=1mm,width=2mm]},color=blue] (2,-0.98)--(2.1,-0.98);
\node[below] at (-0.2,-1) {O};
\node[below] at (4,-1) {$R$};
\node[below] at (-4.2,-1) {$-R$};
\node[above] at (4.5,1) {$R + 2\pi i/\mu$};
\node[above] at (-4.5,1) {$-R + 2\pi i/\mu$};
\node[circle,fill,inner sep=1pt,color=red] at (0,0.3){};
\node[right] at (0,0.3) {$z_3$}; 
\node[circle,fill,inner sep=1pt,color=red] at (0,-0.3){};
\node[left] at (0,-0.3) {$z_2$}; 
\node[right,color=blue] at (4,0) {$\gamma_{\text{\tiny $K$}}$};
\end{tikzpicture}
\end{center}
\caption{A rectangular contour $\gamma_{\text{\tiny $K$}}$ with corners $\pm R$ and $\pm R + 2\pi i/\mu$ determined by a closed, simple, piecewise smooth curve traversed in positive orientation. The function $f$ is holomorphic in an open connected domain containing the interior of $\gamma_{\text{\tiny $K$}}$ and its closure, except at the poles $z = z_2$ and $z = z_3$.}		\label{rectcont} 
\end{figure}
Inside the rectangular contour, the function has two simple poles at the following points:
\begin{align*}
z_2 = \frac{i}{\mu} (\pi - \theta), \qquad \qquad \text{and} \qquad \qquad
z_3 = \frac{i}{\mu} (\pi + \theta)
\end{align*}
that lie on the imaginary axis. The corresponding residue for each of these poles is given as follows:
\begin{align*}
\text{res}(f,z_2) &= \frac{e^{-\frac{k}{\mu} (\pi - \theta)}}{b \mu \sinh i (\pi - \theta)} 
= \frac{e^{-\frac{k}{\mu} (\pi - \theta)}}{i b \mu \sin (\pi - \theta)}
= \frac{e^{-\frac{k}{\mu} (\pi - \theta)}}{i \mu \sqrt{b^2 - c^2}} \\
\text{res}(f,z_3) &= \frac{e^{-\frac{k}{\mu} (\pi + \theta)}}{b \mu \sinh i (\pi + \theta)}  
= \frac{ e^{-\frac{k}{\mu} (\pi + \theta)}}{i b \mu \sin (\pi + \theta)} 
= \frac{-e^{-\frac{k}{\mu} (\pi + \theta)}}{i \mu \sqrt{b^2 - c^2}}.
\end{align*}
Using the Residue Theorem, we know that
\begin{align}
\int_{\gamma_{\text{\tiny $K$}}} f(z) \, dz &= 2 \pi i \sum_{n = 2}^{3} \text{res}(f,z_n)
=\frac{2\pi}{\mu \sqrt{b^2 - c^2}} \left[e^{-\frac{k}{\mu} (\pi - \theta)} - e^{-\frac{k}{\mu} (\pi + \theta)} \right] \nonumber \\
&= \frac{2\pi e^{-\frac{k}{\mu} \pi}}{\mu \sqrt{b^2 - c^2}} \left[e^{\frac{k}{\mu} \theta} - e^{-\frac{k}{\mu} \theta} \right]
= \frac{2\pi \cdot \left(2 e^{-\frac{k}{\mu} \pi} \right)}{\mu \sqrt{b^2 - c^2}} \sinh \left(\frac{k}{\mu} \theta \right). \label{intresi}
\end{align}
On the other hand, expressing the integral on the left-hand side above as a combination of line integrals yields
\begin{align}
\int_{\gamma_{\text{\tiny $K$}}} f(z) \, dz &= 
\int_{-R}^{R} \frac{e^{i k x} \, dx}{b \cosh (\mu x) + c} + 
\int_{0}^{2\pi} \frac{e^{i k (R + iy)} \, i \, dy}{b \cosh \mu (R + iy) + c} \nonumber \\ & \quad - 
\int_{-R}^{R} \frac{e^{i k (x + 2\pi i/\mu)} \, dx}{b \cosh (\mu x + 2\pi i) + c} - 
\int_{0}^{2\pi} \frac{e^{i k (-R + iy)} \, i \, dy}{b \cosh \mu (-R + iy) + c} \nonumber \\
&= \left(1 - e^{-2\frac{k}{\mu} \pi} \right) \int_{-R}^{R} \frac{e^{i k x} \, dx}{b \cosh (\mu x) + c} \nonumber \\ & \quad +
i \int_{0}^{2\pi} \frac{e^{ i k R} e^{-\omega y} \, dy}{b \cosh \mu (R + iy) + c} -  
i \int_{0}^{2\pi} \frac{e^{-i k R} e^{-\omega y} \, dy}{b \cosh \mu (R - iy) + c} \label{intline}
\end{align}
since $\cosh(\mu x + 2\pi i) = \cosh (\mu x) \cos (2\pi) + i \sinh(\mu x) \sin (2\pi) = \cosh (\mu x)$. Furthermore,
$\left| \cosh \mu (\pm R + iy) \right|^2 = \cosh^2(\mu R) - \sin^2 (\mu y) \geq \cosh^2 (\mu R) - 1 \geq \cosh^2 (\mu R)$, $|b \cosh \mu (\pm R + iy)| \geq b \cosh (\mu R)$, and for sufficiently large $R$, $|b \cosh \mu (\pm R + iy) + c| \geq b \cosh (\mu R) - c$. Hence
\begin{align*}
\left|i \int_{0}^{2\pi} \frac{e^{ i k R} e^{-k y} \, dy}{b \cosh \mu (R + iy) + c} \right| &\leq 
\int_{0}^{2\pi} \left|\frac{e^{ i k R} e^{- k y}}{b \cosh \mu (R + iy) + c} \right| \, dy \leq 
\int_{0}^{2\pi} \frac{e^{- k y}}{b \cosh \mu R - c} \, dy \\
&\leq \frac{1 - e^{-2 k \pi}}{k (b \cosh \mu R -c )}
\end{align*}
which tends to $0$ as $R \to \infty$. Similarly, for sufficiently large $R$
\begin{align*}
\left|i \int_{0}^{2\pi} \frac{e^{-i k R} e^{-k y} \, dy}{b \cosh \mu (-R + iy) + c} \right| &\leq 
\int_{0}^{2\pi} \left|\frac{e^{-i k R} e^{- k y}}{b \cosh \mu (-R + iy) + c} \right| \, dy \leq 
\int_{0}^{2\pi} \frac{e^{- k y}}{b \cosh \mu R - c} \, dy \\
& \leq \frac{1 - e^{-2 k \pi}}{k (b \cosh \mu R -c )}
\end{align*}
which again tends to $0$ as $R \to \infty$. Hence, by letting $R \to \infty$, we observe from~\eqref{intline} and~\eqref{intresi} that
\begin{equation*}
\left(1 - e^{-2\frac{k}{\mu} \pi} \right) \int_{-\infty}^{\infty} \frac{e^{i k x} \, dx}{b \cosh (\mu x) + c} = \frac{2\pi \cdot \left(2 e^{-\frac{k}{\mu} \pi} \right)}{\mu \sqrt{b^2 - c^2}} \sinh \left(\frac{k}{\mu} \theta \right). 
\end{equation*}
By taking the real-part of the integral on the left-hand side and pulling out the factor $2 e^{-\frac{k}{\mu} \pi}$ from the product in front of it, we obtain
\begin{align*}
\int_{-\infty}^{\infty} \frac{\cos \left(k x \right) \, dx}{b \cosh (\mu x) + c} 
&= \frac{2\pi}{\mu \sqrt{b^2 - c^2}} \frac{\sinh \left(\frac{k}{\mu} \theta \right)}{\sinh \left(\frac{k}{\mu} \pi \right)}
 = \frac{2\pi}{\mu} \frac{\sinh \left[\frac{k}{\mu} \arccos \left(\frac{c}{b} \right) \right]}{\sqrt{b^2 - c^2} \sinh \left(\frac{k}{\mu} \pi \right)}.
\end{align*}
This completes the proof.
\end{proof}
\begin{figure}[htbp]
\begin{center}
\includegraphics[width=0.45\textwidth]{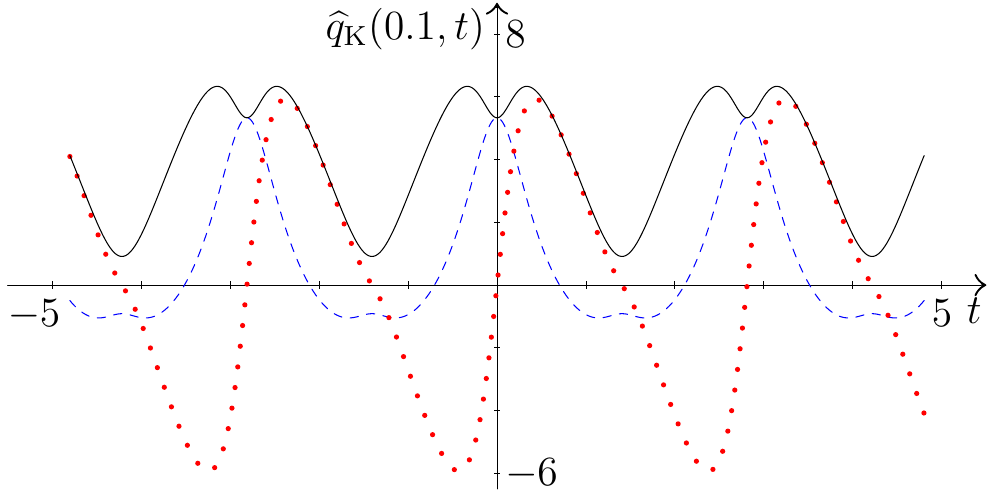} 	\hspace*{1cm}
\includegraphics[width=0.45\textwidth]{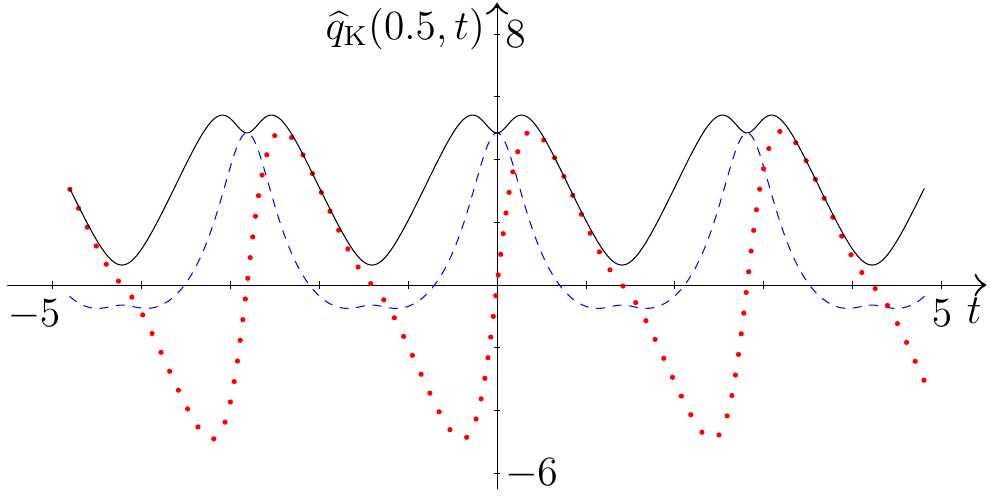} \\ 	\vspace*{0.5cm} 
\includegraphics[width=0.45\textwidth]{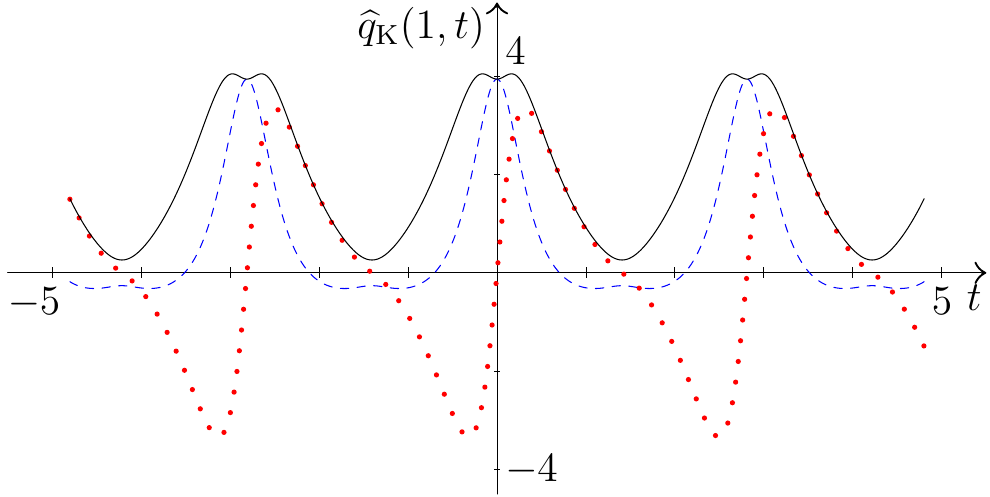} 		\hspace*{1cm}
\includegraphics[width=0.45\textwidth]{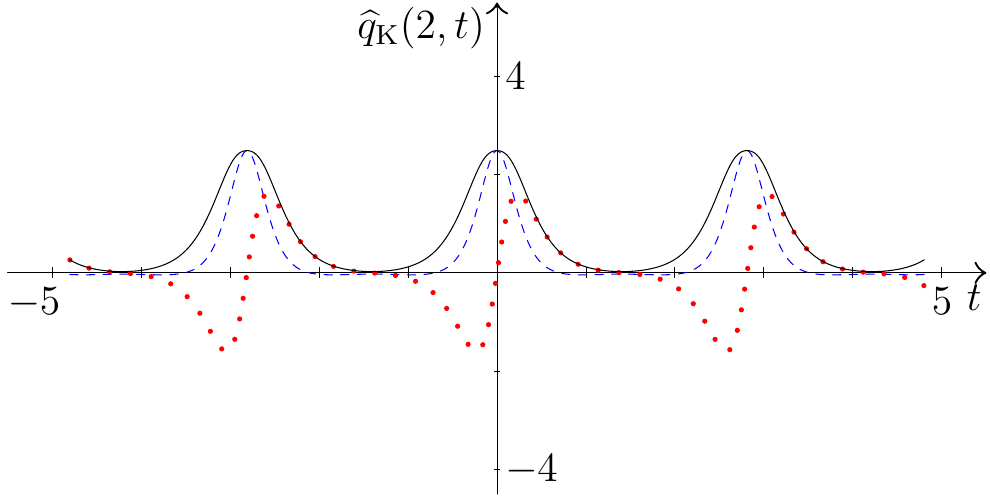}
\end{center}
\caption{The spatial Fourier spectrum evolution for the KM breather $\widehat{q}_K(k_0,t)$ as a function of time for $\mu = 1$, $k_0 = 0.1$ (top left), $k_0 = 0.5$ (top right), $k_0 = 1$ (bottom left), and $k_0 = 2$ (bottom right). The real part (dashed blue), the imaginary part (dotted red), and the modulus (solid black) are depicted separately.}			\label{fou-kmb-wanu-in-time}
\end{figure}

\begin{theorem}
The spatial Fourier spectrum of the KM breather is given by
\begin{align*}
\widehat{q}_K(k,t) &= \left\{
\begin{array}{ll}
2 \pi e^{2it} \left\{ \frac{\left[\mu^3 \cos (\rho t) + i \mu \rho \sin (\rho t) \right])}{\mu^2 \sqrt{\mu^2 + \sin^2(\rho t)}} \cdot \frac{\sinh \left[\frac{k}{\mu} \arccos \left(-2\frac{\mu}{\rho} \cos (\rho t) \right)\right]}{\sinh \left(\frac{k}{\mu} \pi \right)} - \delta (k) \right\}, & \text{for} \; k \neq 0 \\
2 e^{2it} \left\{ \frac{\left[\mu^3 \cos (\rho t) + i \mu \rho \sin (\rho t) \right]}{\mu^2 \sqrt{\mu^2 + \sin^2(\rho t)}} \cdot \arccos \left(-2\frac{\mu }{\rho} \cos(\rho t) \right) - \pi \delta (0) \right\}, \qquad & \text{for} \; k = 0.
\end{array}
\right.
\end{align*}
\end{theorem}
\begin{proof}
Using Definition~\ref{DefSFT} and the fact that $q_K(x,t)$ is an even function with respect to $x$, the spatial Fourier transform of the KM breather $q_K$ can be written as follows:
\begin{align*}
\widehat{q}_K(k,t) &= \int_{-\infty}^{\infty}  e^{2it} \left(\frac{\mu^3 \cos(\rho t) + i \mu \rho \sin (\rho t)}{\rho \cosh(\mu x) - 2\mu \cos(\rho t)} \right) e^{-i k x} \, dx - \int_{-\infty}^{\infty} 1 \,  e^{-i k x} \, dx \\
&= e^{2it} \left[ \mu^3 \cos(\rho t) + i \mu \rho \sin (\rho t) \right] \int_{-\infty}^{\infty}   \left(\frac{\cos (k x)}{\rho \cosh(\mu x) - 2\mu \cos(\rho t)} \right) \, dx - \int_{-\infty}^{\infty}  e^{-i k x} \, dx .
\end{align*}
Using Lemma~\ref{LemDiracDelta}, taking $b = \rho$ and $c = -2 \mu \cos (\rho t)$ in Lemma~\ref{LemGradRyz}, and since $\sqrt{b^2 - c^2} = \mu \sqrt{\mu^2 + \sin^2 (\rho t)}$ for $\mu > 0$, we obtain
\begin{equation*}
\widehat{q}_K(k,t) = 2 \pi e^{2it} \left\{ \frac{\left[\mu^3 \cos (\rho t) + i \mu \rho \sin (\rho t) \right]}{\mu^2 \sqrt{\mu^2 + \sin^2(\rho t)}} \cdot \frac{\sinh \left[\frac{k}{\mu} \arccos \left(-2\frac{\mu}{\rho} \cos (\rho t) \right)\right]}{\sinh \left(\frac{k}{\mu} \pi \right)} - \delta (k) \right\}, \qquad \text{for} \; k \neq 0.
\end{equation*}
By taking the limit of the expressions in Lemma~\ref{LemGradRyz} as $k \to 0$ yields
\begin{equation*}
\int_{-\infty}^{\infty} \frac{dx}{\rho \cosh(\mu x) - 2\mu \cos (\rho t)} = \frac{2}{\mu^2 \sqrt{\mu^2 + \sin^2 (\rho t)}} \arccos \left(-2\frac{\mu }{\rho} \cos (\rho t)\right).
\end{equation*}
Hence
\begin{equation*}
\widehat{q}_K(0,t) = 2 e^{2it} \left\{ \frac{\left[\mu^3 \cos (\rho t) + i \mu \rho \sin (\rho t) \right]}{\mu^2 \sqrt{\mu^2 + \sin^2(\rho t)}} \cdot \arccos \left(-2\frac{\mu}{\rho} \cos(\rho t) \right) - \pi \delta (0) \right\}.
\end{equation*}
This completes the proof. 
\end{proof}
\begin{figure}[htbp]
\begin{center}
\includegraphics[width=0.45\textwidth]{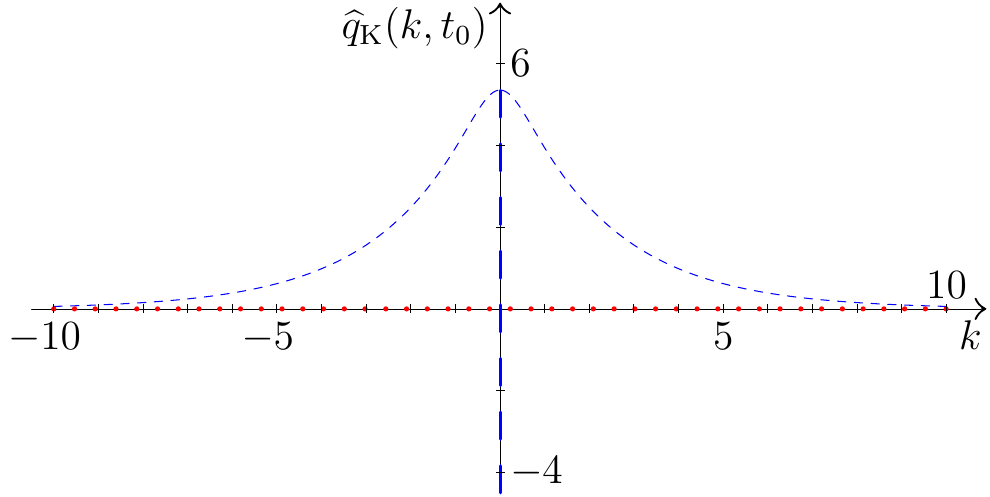} 		\hspace*{1cm}
\includegraphics[width=0.45\textwidth]{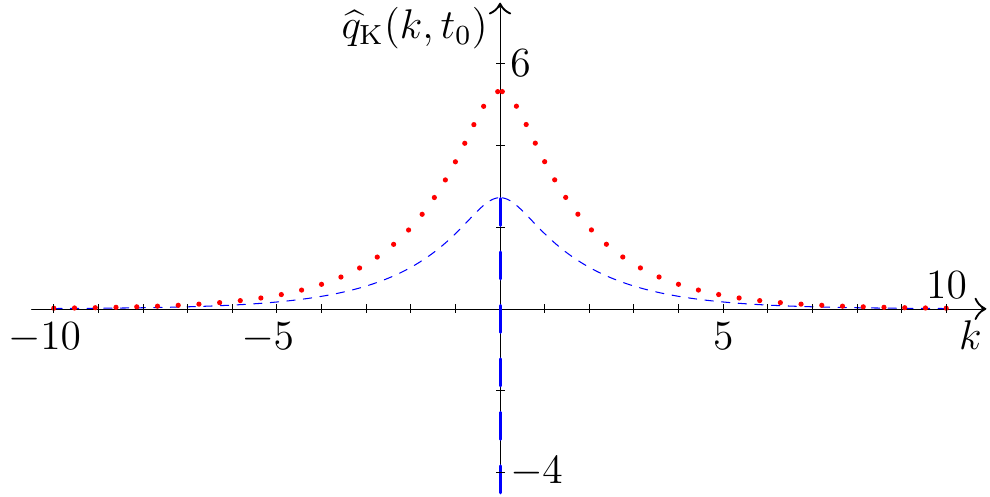} \\ 	\vspace*{0.5cm} 
\includegraphics[width=0.45\textwidth]{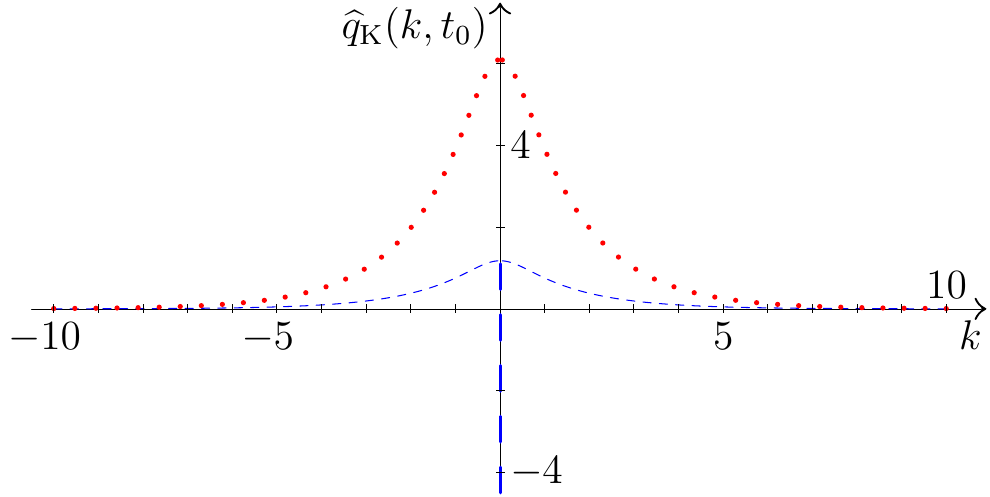} 		\hspace*{1cm}
\includegraphics[width=0.45\textwidth]{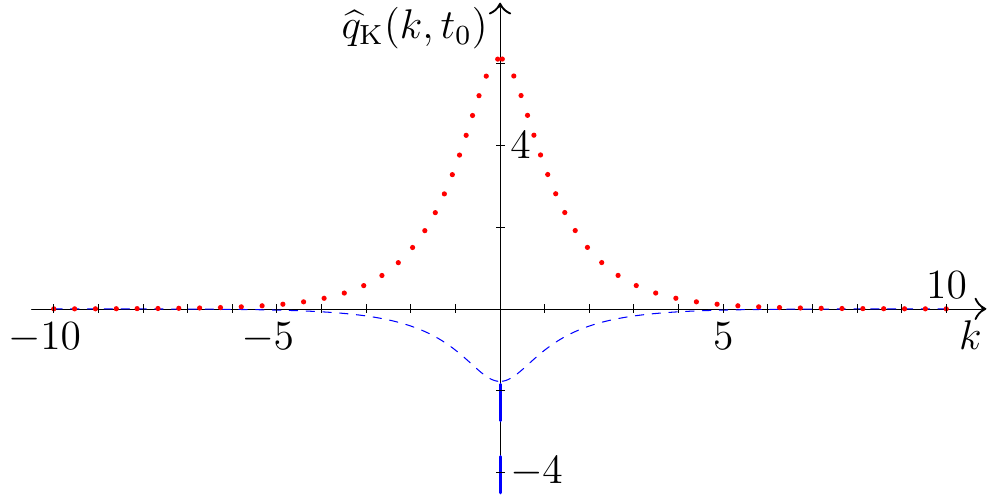}
\end{center}
\caption{The spatial Fourier spectrum of the KM breather $\widehat{q}_K(k,t_0)$ as a function of the wavenumber $k$ for $\mu = 1$ when $t_0 = 0$ (top left), $t_0 = \pi/(8\rho)$ (top right), $t_0 = \pi/(6\rho)$ (bottom left), and $t_0 = \pi/(4\rho)$ (bottom right). The real and imaginary parts are given by dashed blue and dotted red plots, respectively.}		\label{fou-kmb-time-in-wanu}
\end{figure}
Figure~\ref{fou-kmb-wanu-in-time} displays the spatial Fourier spectrum of the KM breather as a function of time for several values of wavenumber $k = k_0$. All panels correspond to the parameter value $\mu = 1$. The real part, imaginary part, and modulus of the spectrum are depicted as dashed blue, dotted red, and solid black curves, respectively. We observe that all spectrum components are periodic with respect to time. For small $k$, the spectrum modulus reaches local maximum values of larger than six while the local minimum values are also positive and larger than one. As the values of $k$ increase, both local maximum and minimum values of the spectrum decrease, where the latter could vanish for particular temporal values. Furthermore, the spectrum moduli feature a local minimum that is sandwiched between local maxima for $k \leq 1$ in the neighborhood when $t = 0$. This feature starts to disappear as the wavenumber increases and eventually, only local maxima are visible.

Figure~\ref{fou-kmb-time-in-wanu} illustrates the spatial Fourier spectrum of the KM breather as a function of the wavenumber for several values of time $t = t_0$. We also take the breather parameter $\mu = 1$ in all cases. The real and imaginary parts of the spectrum are given by dashed blue and dotted red curves, respectively. For the former, the presence of the Dirac delta function is noticeable at $k = 0$, where the curves go to infinity negatively. The spectrum reduces to a real-valued function when $t = 0$ and its corresponding periodic time-scale. Both spectrum components vanish as $k \to \pm \infty$. While the imaginary part always reaches a local maximum at $k = 0$ for $t \neq 0$, the real-part does not always behave like that. As shown at the bottom right panel, the real part may reach a local minimum in a limiting sense when $k \to 0$ at $k = \pi/(4\rho)$. 

\subsection{Peregrine soliton}

Before deriving the spectrum for the Peregrine soliton, we state and verify the following lemma.
\begin{lemma}		\label{LemInt2}
For $a > 0$ and $k \in \mathbb{R}$, we have the following integral:
\begin{equation}
\int_{-\infty}^{\infty} \frac{\cos (k x)}{a + (2x)^2} \, dx = \frac{\pi}{2 \sqrt{a}} e^{-\frac{|k|}{2}\sqrt{a}}.
\end{equation} 	
\end{lemma}
\begin{proof}
We show the proof using the semi-circular contour theorem from complex analysis, e.g., Theorem 9.1 in~\cite{howie2003complex} (on page 154--155). Let
\begin{equation*}
f(z) = \frac{e^{i |k| z}}{a + (2z)^2}
\end{equation*}
be a complex-valued function and meromorphic in the upper half-plane. We observe that $f$ has no poles on the real axis and it has only two poles at $z = i\sqrt{a}/2$ and $z = -i\sqrt{a}/2$. Only the first of these is in the upper half-plane. We calculate that
\begin{equation*}
\text{res}\left(f, i\frac{\sqrt{a}}{2} \right) = \frac{e^{-\frac{|k|}{2}\sqrt{a}}}{4i \sqrt{a}}.
\end{equation*}
In the upper half-plane, for all sufficiently large $|z|$, we have
\begin{align*}
\left|z f(z) \right| = \left|\frac{z \, e^{i |k| z}}{a + (2z)^2} \right| = \left|\frac{z \, e^{i |k| x} \, e^{-|k|y}}{4z^2 + a} \right| \leq 
\frac{|z|}{4 |z|^2 - a}.
\end{align*}

\begin{figure}[htbp]
	\begin{center}
		\begin{tikzpicture}
		\draw[->] (-5,0)--(5,0) node[right]{Re($z$)};
		\draw[->] (0,-1)--(0,4) node[above]{Im($z$)};
		\draw[color=blue,thick] (3,0) arc (0:180:3) to cycle;
		\node[below] at (-0.2,0) {O};
		\node[below] at (3,0) {$R$};
		\node[below] at (-3.1,0) {$-R$};
		\draw[-{Stealth[length=1mm,width=2mm]},color=blue] (-1.6,0)--(-1.5,0);
		\draw[-{Stealth[length=1mm,width=2mm]},color=blue] (1.5,0)--(1.6,0);
		\draw[-{Stealth[length=1mm,width=2mm]},color=blue] (2,2.23)--(1.95,2.26);
		\draw[-{Stealth[length=1mm,width=2mm]},color=blue] (-1.95,2.27)--(-2,2.23);
		\node[circle,fill,inner sep=1pt,color=red] at (0,2){};
		\node[right] at (0,2) {$i\frac{\sqrt{a}}{2}$};
		\node[color=blue] at (3,1.3) {$\gamma_{\text{\tiny $P$}}$};
		\end{tikzpicture}
	\end{center}
	\caption{A semicircle contour $\gamma_{\text{\tiny $P$}} = \left[-R,R \right] \cup \left\{z : |z| = R, \; \text{Im} \; z \geq 0 \right\}$, traversed in the positive direction with $R > \sqrt{a}/2$.} 		\label{semicircle}
\end{figure}
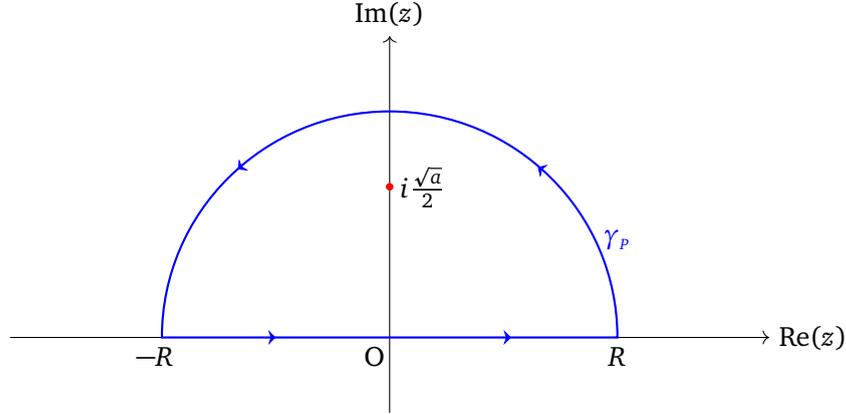
We have adopted the relationship Re$(z) = x$ and Im$(z) = y$ and the fact that $\left|e^{i |k| x} \right| = 1$ and $\left|e^{-|k|y} \right| \leq 1$ for $y \geq 0$. It follows that $\left|z f(z)\right|$ tends uniformly to $0$ as $|z| \to \infty$ in the upper half-plane. Let $\gamma_{\text{\tiny $P$}}$ be a semicircle contour given by $\gamma_{\text{\tiny $P$}} : \left[-R,R \right] \cup \left\{z : |z| = R, \; \text{Im} \; z \geq 0 \right\}$, traversed in the positive direction with $R > \sqrt{a}/2$, as shown in Figure~\ref{semicircle}, then
\begin{align*}
\text{(PV)} \int_{-\infty}^{\infty} f(x) \, dx &= \lim\limits_{R \to \infty} \int_{-R}^{R} f(x) \, dx 
= \int_{\gamma_{\text{\tiny $P$}}} f(z) \, dz - \lim\limits_{R \to \infty} \int_{0}^{\pi} f(R e^{i \theta}) \, i R e^{i \theta} \, d\theta \\
\end{align*}
where PV stands for the principal value. Indeed, from the condition of $\left|z f(z) \right|$ uniformly in the upper half-plane as $|z| \to \infty$, it implies that for all $\varepsilon > 0$, there exists $S > 0$ such that $\left|z f(z) \right| < \varepsilon$ for all $z$ in the upper half-plane such that $|z| > S$. Consequently, for all $R > S$
\begin{equation*}
\left| \int_{0}^{\pi} f(R e^{i \theta}) \, i R e^{i \theta} \, d\theta \right| \leq \int_{0}^{\pi} \left|\frac{e^{i |k| R e^{i \theta}}  i R e^{i \theta}}{a + 4 R^2 e^{2i \theta}}\right| \, d\theta \leq \frac{\pi R}{4 R^2 - a} < \pi \varepsilon
\end{equation*}
and thus
\begin{equation*}
\int_{0}^{\pi} f(R e^{i \theta}) \, i R e^{i \theta} \, d\theta \to 0 \qquad \text{as} \quad R \to \infty.
\end{equation*}
Since there exists a constant $\widehat{S} > 0$ such that $\left|f(x)\right| \leq \widehat{S}/x^2$ for large $|x|$, then ${\displaystyle \int_{\infty}^{\infty} f(x) \, dx}$ exists and equals to its Cauchy principal value. Hence, we can dispense with the principal value and by applying the Residue Theorem, we deduce that
\begin{equation*}
\int_{-\infty}^{\infty} \frac{e^{i |k| x}}{a + (2x)^2} \, dx = 
\int_{-\infty}^{\infty} \frac{\cos (kx) + i \sin \left(|k|x \right)}{a + (2x)^2} \, dx =
2 \pi i \; \text{res}\left(f, i\frac{\sqrt{a}}{2} \right) = \frac{\pi}{2 \sqrt{a}} e^{-\frac{|k|}{2} \sqrt{a}}.
\end{equation*}
Equating the real parts gives the desired result and the proof is completed.
\end{proof}
\begin{figure}[htbp]
	\begin{center}
		\includegraphics[width=0.45\textwidth]{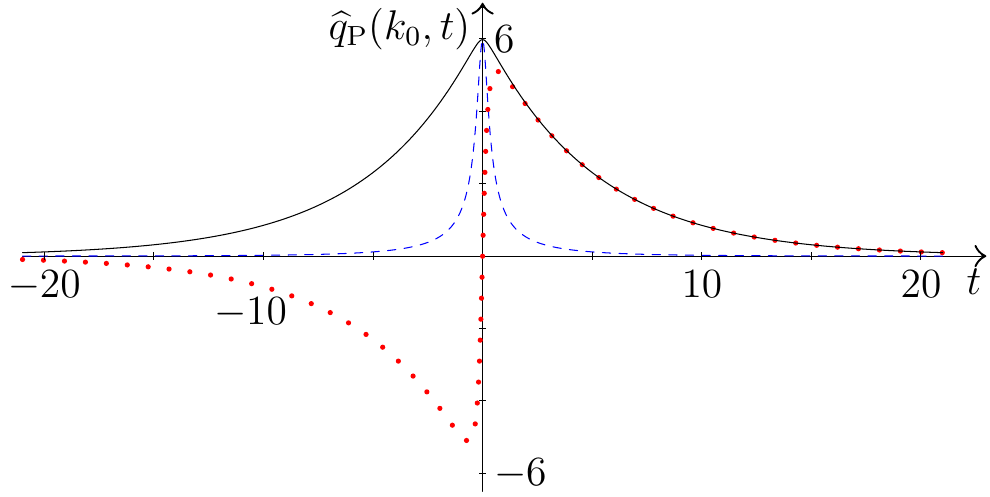} 		\hspace*{1cm}
		\includegraphics[width=0.45\textwidth]{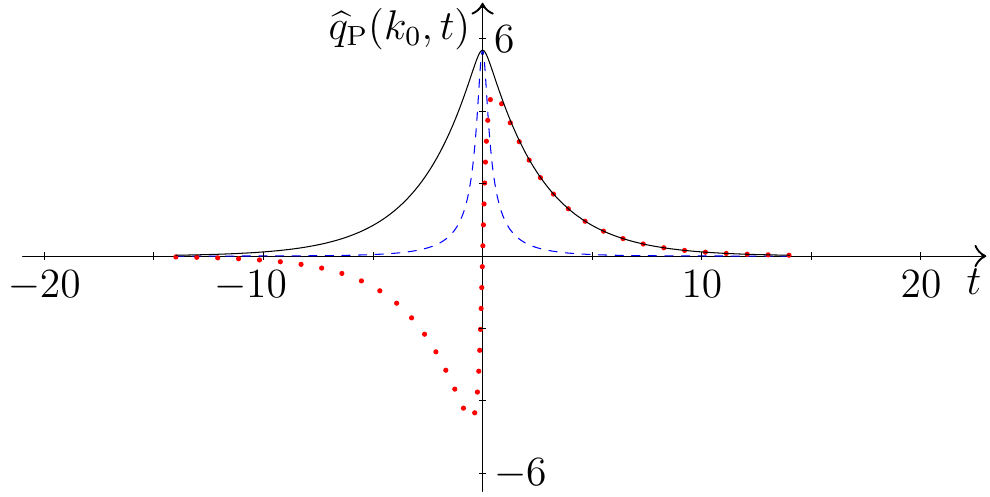} \\ 	\vspace*{0.5cm} 
		\includegraphics[width=0.45\textwidth]{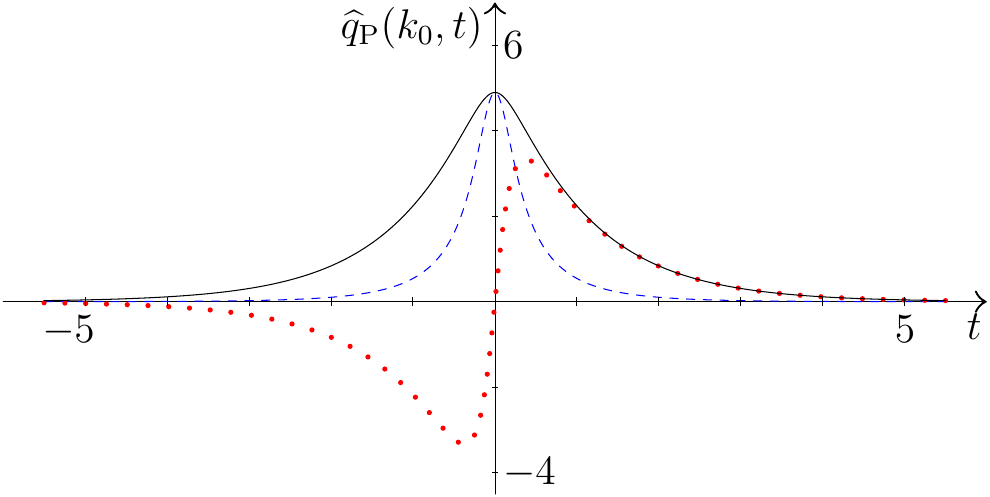} 		\hspace*{1cm}
		\includegraphics[width=0.45\textwidth]{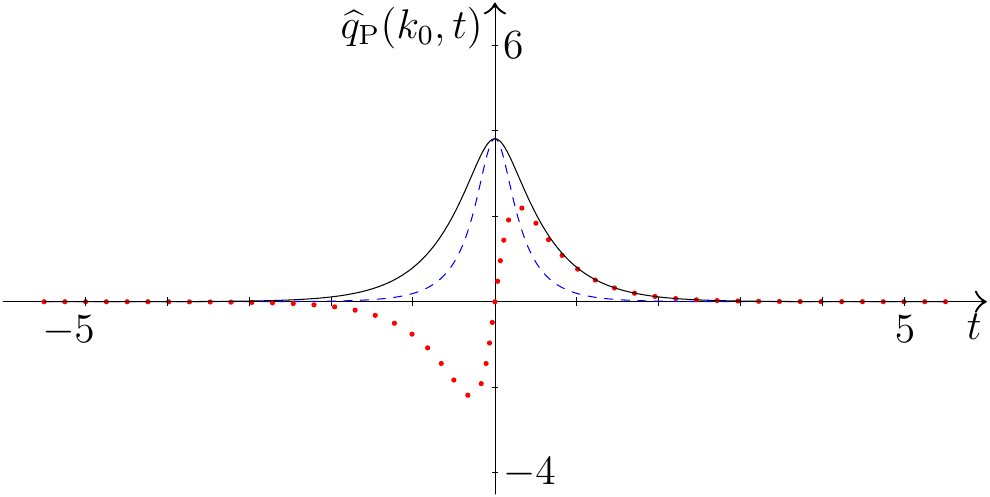}
	\end{center}
	\caption{The spatial Fourier spectrum of the Peregrine breather $\widehat{q}_P(k_0,t)$ as a function of time $t$ for $k_0 = 0.1$ (top left), $k_0 = 0.2$ (top right), $k_0 = 0.5$ (bottom left), and $k_0 = 1$ (bottom right). The real and imaginary parts are given by dashed blue and solid red plots, respectively. The solid black curve is the modulus spectrum amplitude.}		\label{fou-per-wanu-in-time}
\end{figure}

\begin{theorem}
The spatial Fourier spectrum of the Peregrine soliton $q_P(x,t)$ is given by
\begin{equation}
\widehat{q}_P(k,t) = 2\pi e^{2it} \left(\frac{1 + 4it}{\sqrt{1 + (4t)^2}} e^{-\frac{k}{2} \sqrt{1 + (4t)^2}} - \delta (k) \right).
\end{equation}
\end{theorem}
\begin{proof}
Using Definition~\eqref{DefSFT} and the fact that $q_P(x,t)$ is an even function with respect to $x$, we can write the spatial Fourier transform $\widehat{q}_P(k,t)$ as follows:
\begin{align}
\widehat{q}_P(k,t) &= \int_{-\infty}^{\infty} e^{2it} \left(\frac{4 (1 + 4 i t)}{1 + (4t)^2 + (2x)^2} \right) e^{-i k x} \, dx - \int_{-\infty}^{\infty} 1 \, e^{-i k x} \, dx \\
&= 4 e^{2it} (1 + 4it) \int_{-\infty}^{\infty} \frac{\cos (k x)}{1 + (4t)^2 + (2x)^2} \, dx - \int_{-\infty}^{\infty} e^{-i k x} \, dx. 
\end{align}	
Using Lemmas~\ref{LemDiracDelta} and~\ref{LemInt2}, we obtain the desired result:
\begin{align}
\widehat{q}_P(k,t) &= 4 e^{2it} (1 + 4it) \left(\frac{\pi}{2 \sqrt{1 + (4t)^2}} e^{-\frac{k}{2} \sqrt{1 + (4t)^2}} \right) - 2\pi \delta(k) \\
&= 2\pi e^{2it} \left(\frac{1 + 4it}{\sqrt{1 + (4t)^2}} e^{-\frac{k}{2} \sqrt{1 + (4t)^2}} - \delta (k) \right).
\end{align}
This completes the proof of the theorem.
\end{proof}
\begin{figure}[htbp]
\begin{center}
\includegraphics[width=0.45\textwidth]{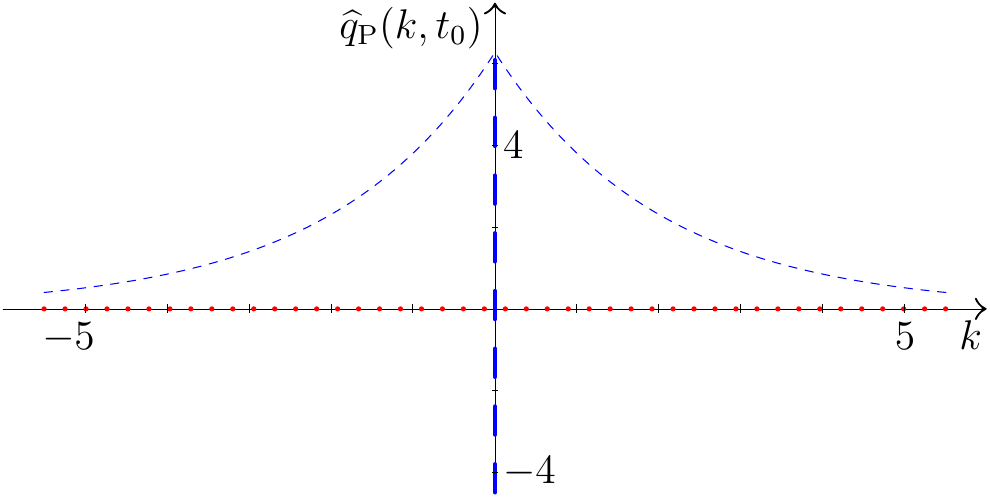} 		\hspace*{1cm}
\includegraphics[width=0.45\textwidth]{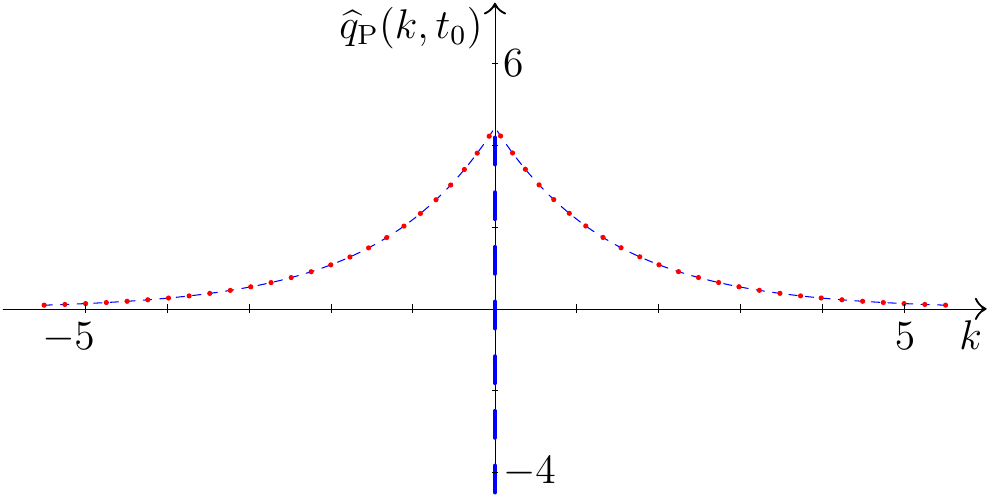} \\ 	\vspace*{0.5cm} 
\includegraphics[width=0.45\textwidth]{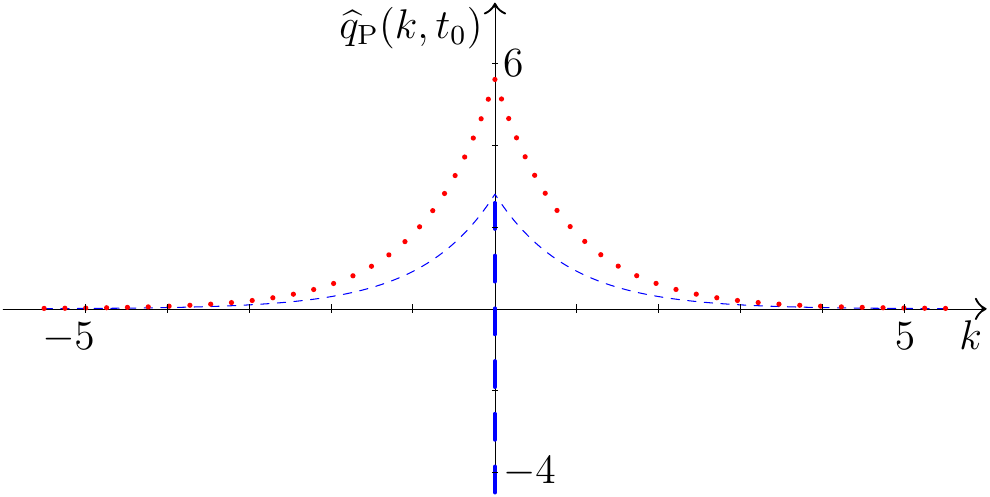} 		\hspace*{1cm}
\includegraphics[width=0.45\textwidth]{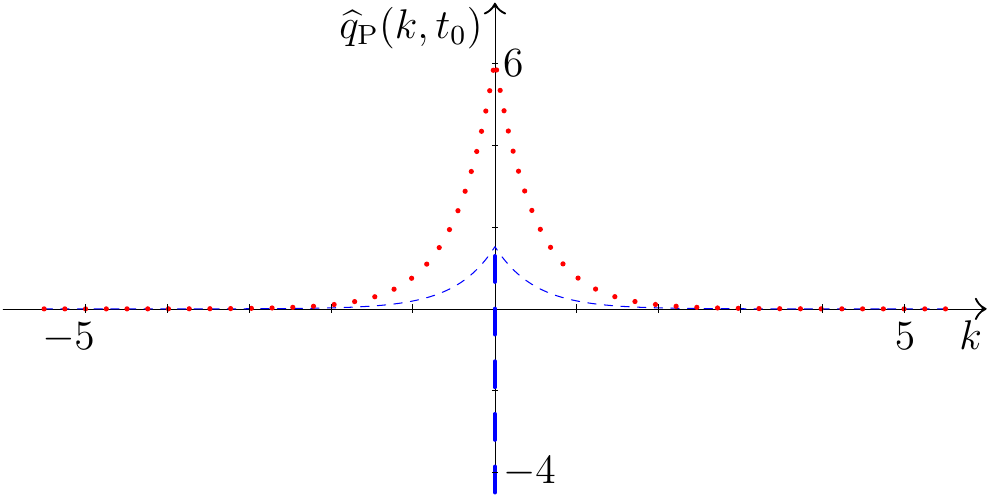}
\end{center}
\caption{The spatial Fourier spectrum of the Peregrine breather $\widehat{q}_P(k,t_0)$ as a function of the wavenumber $k$ when $t_0 = 0$ (top left), $t_0 = 1/4$ (top right), $t_0 = 1/2$ (bottom left), and $t_0 = 1$ (bottom right). The real and imaginary parts are given by dashed blue and dotted red plots, respectively.}		\label{fou-per-time-in-wanu}
\end{figure}
Figure~\ref{fou-per-wanu-in-time} shows the spatial Fourier spectrum of the Peregrine breather as a function of time for various values of wavenumber $k = k_0$. The real-part, imaginary part, and modulus of the spectrum are depicted in dashed blue, dotted red, and solid black curves, respectively. All spectrum components tend to vanish as $t \to \pm \infty$. Both the real-part and modulus of the spectrum reach local maxima at $t = 0$ while the imaginary part takes negative and positive values for $t < 0$ and $t > 0$, respectively. It reaches local minimum and maximum for some values $t$ not too far away from $t = 0$. The value of local maxima for the spectrum modulus decreases as the value of the wavenumber increases. For small wavenumber, the local maximum reaches the value of six as $k \to 0$.

Figure~\ref{fou-per-time-in-wanu} features the spatial Fourier spectrum of the Peregrine breather as a function of the wavenumber for several values of time $t = t_0$. The real and imaginary parts of the spectrum are plotted as dashed blue and dotted red curves, respectively. Similar to the previous case, the appearance of the Dirac delta function is visible at $k = 0$ as the curve extends infinitely negative. The imaginary part vanishes when $t_0 = 0$ and thus the spectrum reduces to a real-valued function. For $t > 0$, both real and imaginary parts tend to vanish as $k \to \pm \infty$. They also reach local maxima at $k = 0$, and depending on the values of $k$, the local maxima for the real part can be larger, equal, or smaller than the ones for the imaginary part. The largest of these quantities takes the value of slightly larger than six at $k = 0$.

\subsection{Akhmediev breather}

\begin{lemma}	\label{lemtrigint}
For $0 < |a| < b$, the following trigonometric integral holds, c.f.~Formula~3.613.1 in~\cite{gradshteyn2015table}:
\begin{equation}
\int_{0}^{2\pi} \frac{\cos (n \xi)\, d \xi}{b - a\,\cos \xi} = \frac{2 \pi}{\sqrt{b^{2} - a^{2}}} \left(\frac{b - \sqrt{b^{2} - a^{2}}}{a} \right)^{n}, \qquad \; \textmd{for} \; \; n \in \mathbb{Z}.		  \label{trigoint0}
\end{equation}
\end{lemma}
\begin{proof}
By performing the integration in the complex plane, we can write the integral as the real part of another integral expressed on the right-hand side of the following equation:
\begin{equation*}
\int_{0}^{2\pi} \frac{\cos (n \xi)\, d \xi}{b - a\,\cos \xi} = \text{Re} \left\{\int_{0}^{2\pi} \frac{e^{i n \xi} \, d \xi}{b - a\,\cos \xi} \right\}.
\end{equation*}
Let $C: z = e^{i \xi}$, $0 \leq \xi \leq 2\pi$ be the unit circle in the complex plane, then $e^{i n \xi} = z^n$, $\cos \xi = \frac{1}{2}(z + z^{-1})$, and $d\xi = -i dz/z$. See Figure~\ref{fullcircle}. Substituting these into the integral above, we can express it as
\begin{equation*}
\int_{0}^{2\pi} \frac{\cos (n \xi)\, d \xi}{b - a\,\cos \xi} = 2i \oint_{C} \frac{z^{n} \, dz}{a z^2 - 2 b z + a}
= 2i \oint_{C} \frac{z^{n} \, dz}{a (z - z_1)(z - z_2)}
\end{equation*}
where
\begin{equation*}
z_1 = \frac{b + \sqrt{b^2 - a^2}}{a} \qquad \text{and} \qquad z_2 = \frac{b - \sqrt{b^2 - a^2}}{a}
\end{equation*}
are the roots of the quadratic equation $a z^2 - 2 b z + a = 0$. The integrand has simple poles at $z = z_1$ and $z = z_2$. Observe that
\begin{equation*}
|z_1| = \frac{b + \sqrt{b^2 - a^2}}{a}  > \frac{b}{a} > 1 
\end{equation*}
and thus $z_1$ does not lie within the contour of $C$. Furthermore, since $z_1 z_2 = a/a = 1$, we deduce that $|z_2| < 1$. The residue at $z = z_2$ is given by
\begin{equation*}
\frac{z_2^n}{a(z_2 - z_1)} = \frac{1}{-2 \sqrt{b^2- a^2}} \left(\frac{b + \sqrt{b^2 - a^2}}{a}\right)^n.
\end{equation*}
Using Cauchy's Residue Theorem, we obtain the desired trigonometric integral:
\begin{equation*}
\int_{0}^{2\pi} \frac{\cos (n \xi)\, d \xi}{b - a\,\cos \xi} =  \frac{(2 \pi i) \, (2i)}{-2 \sqrt{b^2- a^2}} \left(\frac{b + \sqrt{b^2 - a^2}}{a}\right)^n
= \frac{2 \pi}{\sqrt{b^{2} - a^{2}}} \left(\frac{b - \sqrt{b^{2} - a^{2}}}{a} \right)^{n}.
\end{equation*}
This completes the proof.
\end{proof}
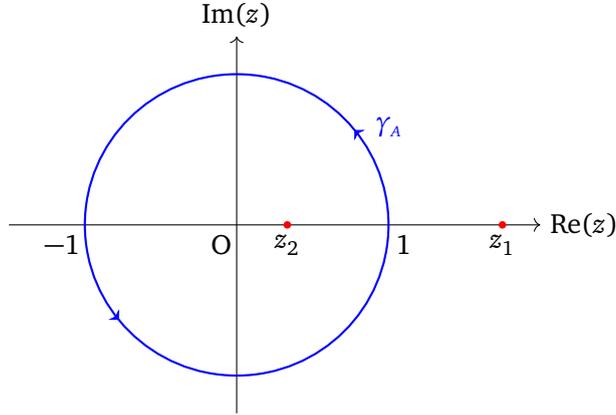
\begin{figure}[htbp]
\begin{center}
\begin{tikzpicture}
\draw[->] (-3,0)--(4,0) node[right]{Re($z$)};
\draw[->] (0,-2.5)--(0,2.5) node[above]{Im($z$)};
\draw[color=blue,thick] (2,0) arc (0:360:2) to cycle;
\node[below] at (-0.2,0) {O};
\node[below] at (2.2,0) {$1$};
\node[below] at (-2.3,0) {$-1$};
\draw[-{Stealth[length=1mm,width=2mm]},color=blue] (1.6,1.21)--(1.55,1.26);
\draw[-{Stealth[length=1mm,width=2mm]},color=blue] (-1.6,-1.21)--(-1.55,-1.26);
\node[circle,fill,inner sep=1pt,color=red] at (2/3,0){};
\node[below] at (2/3,0) {$z_2$};
\node[circle,fill,inner sep=1pt,color=red] at (3.5,0){};
\node[below] at (3.5,0) {$z_1$};
\node[color=blue] at (2,1.3) {$\gamma_{\text{\tiny $A$}}$};
\end{tikzpicture}
\end{center}
\caption{A unit circle contour $\gamma_{\text{\tiny $A$}} = \left\{z : |z| = 1 \right\}$ traversed in the positive direction with $z_2 < 1 < z_1$.} 		\label{fullcircle}
\end{figure}

\begin{corollary}	\label{cortrig}
For $0 \leq |a| < 1$, we have the following integral:  
\begin{equation}
\int_{0}^{2\pi} \frac{d\xi}{1 - a\,\cos \xi} = \frac{2\pi}{\sqrt{1 - a^{2}}}. 	\label{trigoint1}
\end{equation}	
\end{corollary}
\begin{proof}
By taking $b = 1$ and $n = 0$ in Lemma~\ref{lemtrigint}, i.e., Equation~\eqref{trigoint0}, we obtain the result immediately.
\end{proof}

In what follows, we state and prove the spatial Fourier spectrum for the Akhmediev breather. A similar expression of this spectrum is available in~\cite{akhmediev1997solitons}, but the authors did not provide any proof. On the other hand, a derivation proof using trigonometric integrals and series has been attempted by the author in an appendix of his PhD thesis~\cite{karjanto2006thesis}.
\begin{theorem}
The spatial Fourier amplitude spectrum for the Akhmediev breather $q_A(x,t)$ is given by
\begin{equation}
\widehat{q}_{A}^{(n)}(k, t) = \widehat{q}_{A}^{(n)}(\kappa, t) \left\{
\begin{array}{ll}
e^{2it} \left(\frac{\kappa^3 \cosh(\sigma t) + i \kappa \sigma \sinh (\sigma t)}{\sqrt{4\kappa^2 \cosh^2(\sigma t) - \sigma^2}} - 1 \right), & \quad \text{for} \; n = 0, \\
e^{2it} \left(\frac{\kappa^3 \cosh(\sigma t) + i \kappa \sigma \sinh (\sigma t)}{\sqrt{4\kappa^2 \cosh^2(\sigma t) - \sigma^2}} \right)
\left(\frac{2\kappa \cosh(\sigma t) - \sqrt{4\kappa^2 \cosh^2 (\sigma t) - \sigma^2}}{\sigma}\right)^{n}, & \quad \text{for} \; n \in \mathbb{N}.
\end{array}
\right.
\end{equation}
\end{theorem}

\begin{proof}
The Akhmediev breather is periodic in space with the period $L = 2\pi/\kappa$. Furthermore, we consider the wavenumber values $k_n$ at a discrete multiple of the modulation wavenumber $\kappa$, i.e., $k_n = n \kappa$, $n \in \mathbb{Z}$. Hence,
\begin{equation*}
\widehat{q}_A^{(n)}(\kappa,t) = \frac{\kappa}{4\pi} \int_{-2\pi/\kappa}^{2\pi/\kappa} q_A(x,t) e^{-in \kappa x} \, dx.
\end{equation*}
Employing the variable substitution $\xi = \kappa x$, and using the fact that $q_A$ is an even function with respect to the spatial variable $x$, it reduces to
\begin{align*}
\widehat{q}_A^{(n)}(\kappa,t) &= \frac{1}{4\pi} \int_{-2\pi}^{2\pi} e^{2it} \left(\frac{\kappa^3 \cosh(\sigma t) + i \kappa \sigma \sinh (\sigma t)}{2 \kappa \cosh(\sigma t) - \sigma \cos \xi} - 1 \right) \cos (n\xi) \, d\xi \\
&= \frac{1}{2\pi} \int_{0}^{2\pi} e^{2it} \left(\frac{\kappa^3 \cosh(\sigma t) + i \kappa \sigma \sinh (\sigma t)}{2 \kappa \cosh(\sigma t) - \sigma \cos \xi} - 1 \right) \cos (n\xi) \, d\xi 
\end{align*}
For $n = 0$, applying Lemma~\ref{lemtrigint} or Corollary~\ref{cortrig}, we obtain
\begin{equation*}
\widehat{q}_A^{(0)}(\kappa,t) = e^{2it} \left(\frac{\kappa^3 \cosh(\sigma t) + i \kappa \sigma \sinh (\sigma t)}{\sqrt{4 \kappa^2 \cosh^2(\sigma t) - \sigma^2}} - 1 \right).
\end{equation*} 
For $n \neq 0$, taking $a = \sigma$ and $b = 2 \kappa \cosh(\sigma t)$, we obtain
\begin{equation*}
\widehat{q}_A^{(n)}(\kappa,t) = e^{2it} \left(\frac{\kappa^3 \cosh(\sigma t) + i \kappa \sigma \sinh (\sigma t)}{\sqrt{4 \kappa^2 \cosh^2(\sigma t) - \sigma^2}} \right) \left(\frac{2 \kappa \cosh(\sigma t) - \sqrt{4 \kappa^2 \cosh^2(\sigma t) - \sigma^2}}{\sigma} \right)^n.
\end{equation*} 
This completes the proof.	
\end{proof}
\begin{figure}[htbp]
\begin{center}
\includegraphics[width=0.45\textwidth]{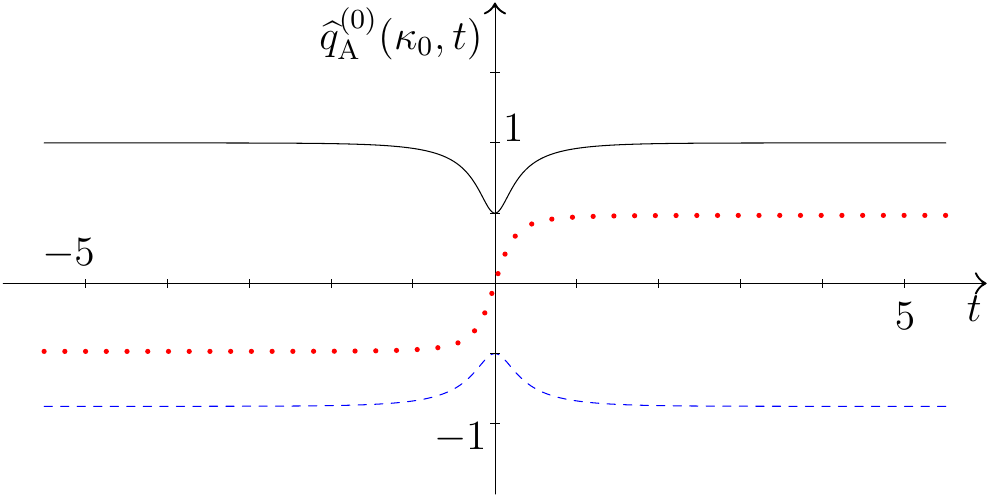} 		\hspace*{1cm}
\includegraphics[width=0.45\textwidth]{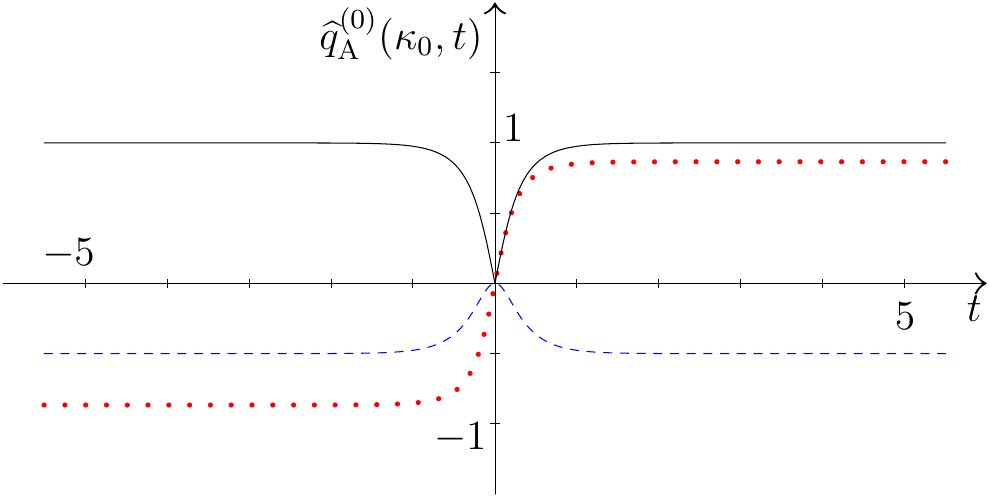} \\ 	\vspace*{0.5cm} 
\includegraphics[width=0.45\textwidth]{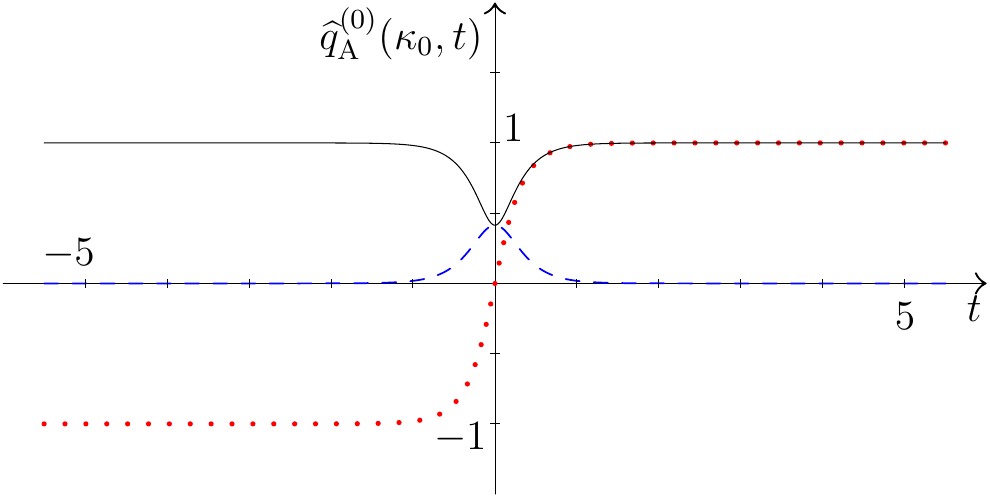} 		\hspace*{1cm}
\includegraphics[width=0.45\textwidth]{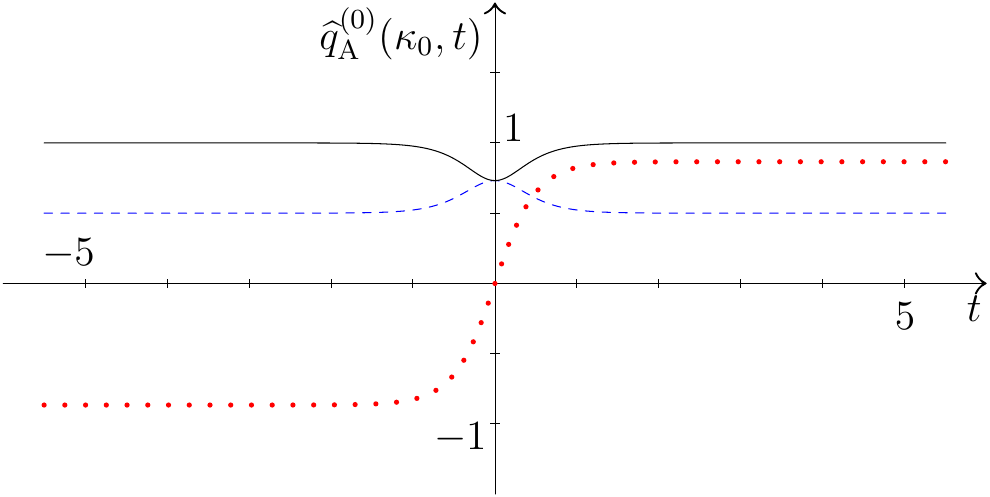}
\end{center}
\caption{The spatial Fourier spectrum of the Akhmediev breather $\widehat{q}_A^{(0)}(\kappa_0,t)$ as a function of time $t$ for $\kappa_0 = 0.5$ (top left), $\kappa_0 = 1$ (top right), $\kappa_0 = \sqrt{2}$ (bottom left), and $\kappa_0 = \sqrt{3}$ (bottom right). The real and imaginary parts are given by dashed blue and solid red plots, respectively. The solid black curve is the modulus spectrum amplitude.}	\label{fou-akh-wanu-in-time}
\end{figure}
Figure~\ref{fou-akh-wanu-in-time} displays the spatial Fourier spectrum of the Akhmediev breather for $n = 0$ as a function of time $t$ for various values of the modulation wavenumber $\kappa = \kappa_0$, i.e., $\widehat{q}_A^{(0)}(\kappa_0,t)$. The real and imaginary parts of $\widehat{q}_A^{(0)}$ are depicted by the dashed blue and dotted red curves, respectively, while the solid black curve represents the modulus of the amplitude spectrum evolution. For all values of $\kappa$, the spectrum modulus reaches a minimum value during $t = 0$, and for an increasing value of $\kappa$, these minimum values of the modulus are decreasing until they vanish at $\kappa_0 = 1$ and then increase again towards one as $\kappa_0 \to 2$. For $t \to \pm \infty$, the amplitude modulus also goes to one.

\begin{figure}[htbp]
\begin{center}
\includegraphics[width=0.45\textwidth]{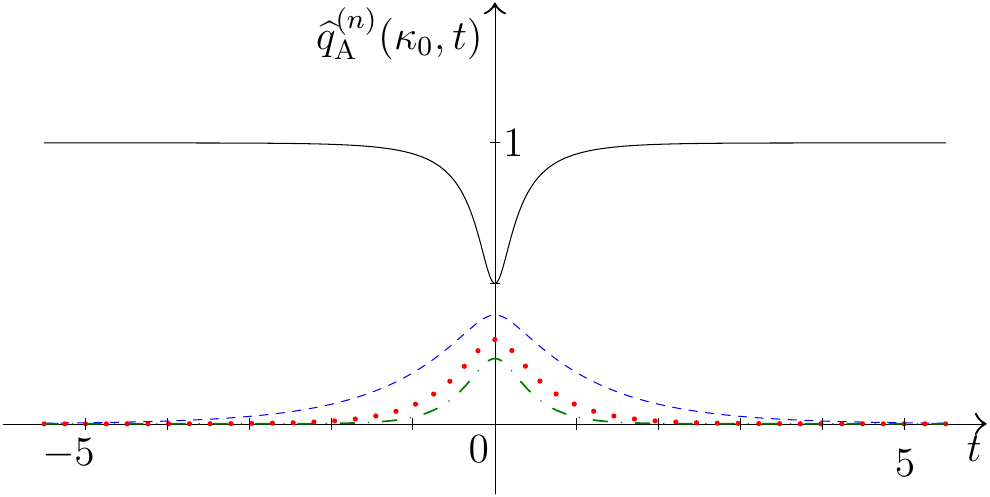} 		\hspace*{1cm}
\includegraphics[width=0.45\textwidth]{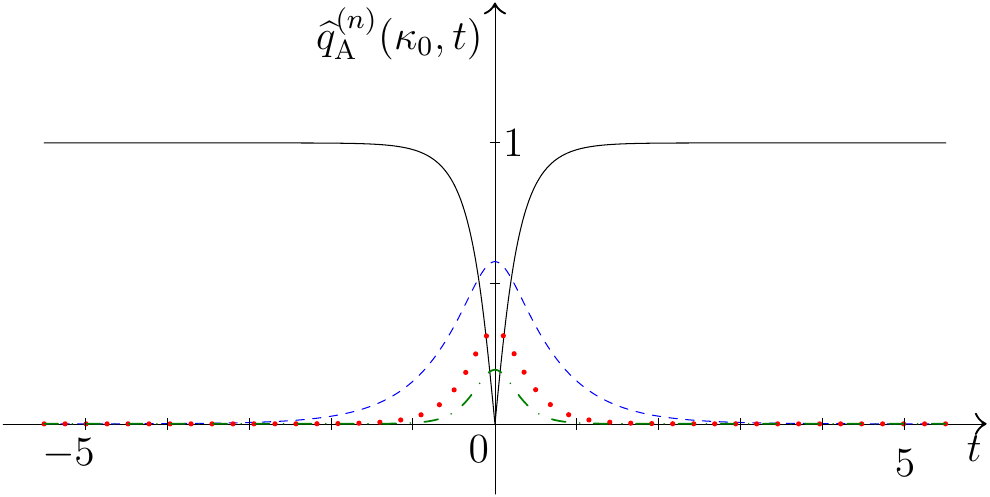} \\ 	\vspace*{0.5cm} 
\includegraphics[width=0.45\textwidth]{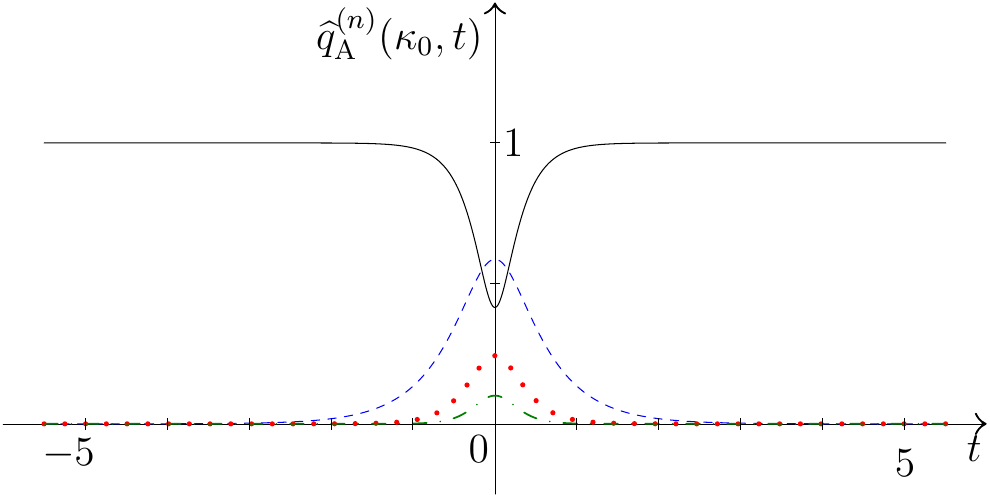} 		\hspace*{1cm}
\includegraphics[width=0.45\textwidth]{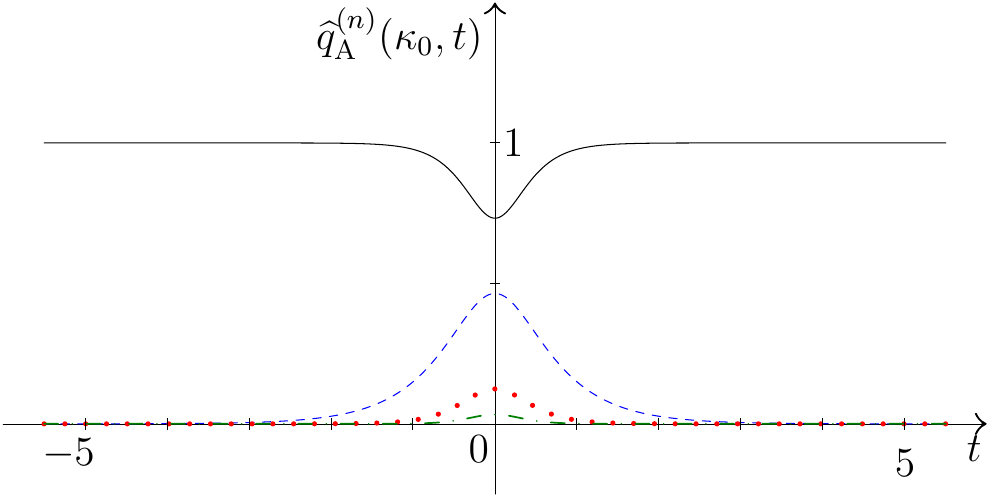}
\end{center}
\caption{The modulus of spatial Fourier amplitude spectrum of the Akhmediev breather $\left|\widehat{q}_A^{(n)}(\kappa_0,t) \right|$ as a function of time $t$ for $\kappa_0 = 0.5$ (top left), $\kappa_0 = 1$ (top right), $\kappa_0 = \sqrt{2}$ (bottom left), and $\kappa_0 = \sqrt{3}$ (bottom right) for $n = 0$ (solid black), $n = 1$ (dashed blue), $n = 2$ (dotted red), and $n = 3$ (dash-dotted green).}	\label{fou-akh-modu-kappa-in-time}
\end{figure}
Similar to Figure~\ref{fou-akh-wanu-in-time}, Figure~\ref{fou-akh-modu-kappa-in-time} displays the moduli of amplitude spectra for several values of~$n$ and~$\kappa$. The solid black curves correspond to the moduli for the main/central wavenumber when $n = 0$, the dashed blue curves correspond to the first pair of sidebands ($n = 1$), the dotted red curves correspond to the second pair of sidebands ($n = 2$), and the dash-dotted green curves correspond to the third pair of sidebands ($n = 3$). While the moduli for the central wavenumber reach a local minimum when $t = 0$, all the sidebands moduli reach a local maximum when $t = 0$. The further the sideband pairs from the main wavenumber, the smaller they attain the maximum value. For $t \to \pm \infty$, all sideband moduli tend to zero while the central sideband goes to one.

\begin{figure}[htbp]
\begin{center}
\includegraphics[width=0.45\textwidth]{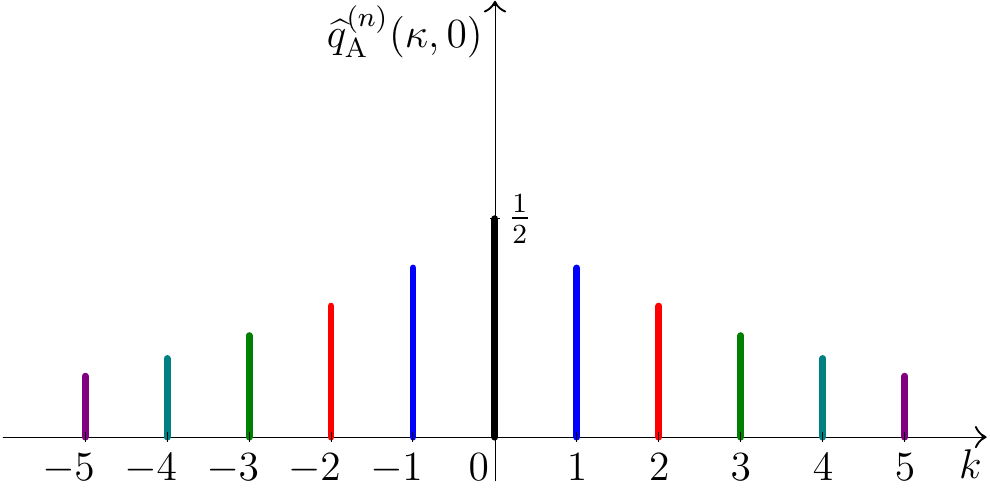} 		\hspace*{1cm}
\includegraphics[width=0.45\textwidth]{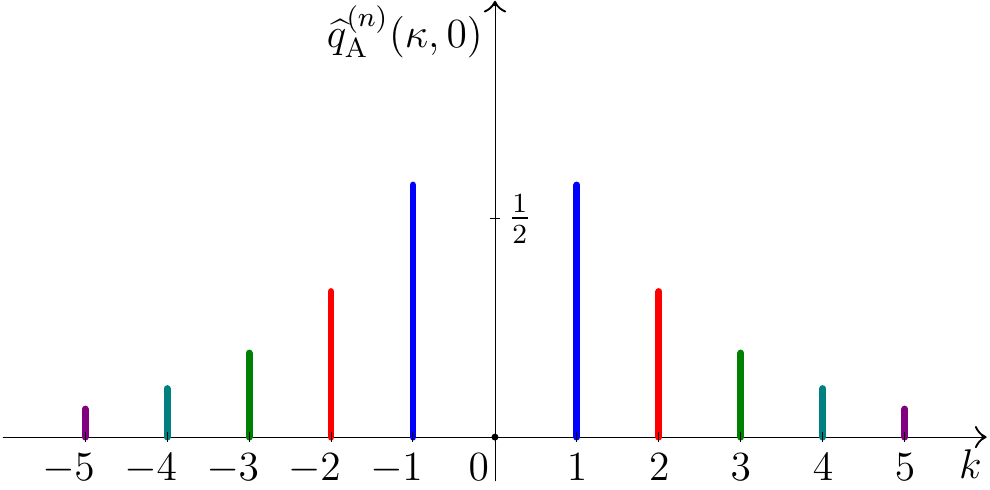} \\ 	\vspace*{0.5cm} 
\includegraphics[width=0.45\textwidth]{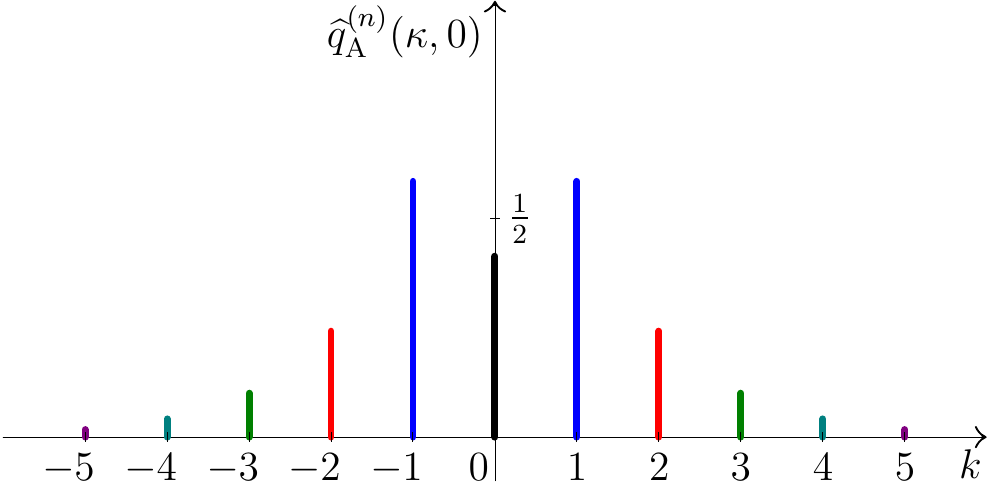} 		\hspace*{1cm}
\includegraphics[width=0.45\textwidth]{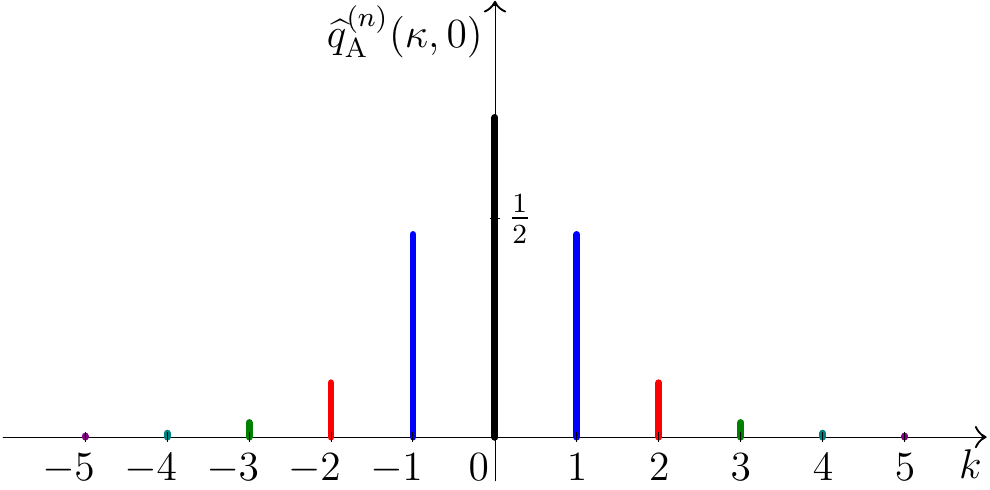}
\end{center}
\caption{The modulus of spatial Fourier amplitude spectrum of the Akhmediev breather $\left|\widehat{q}_A^{(n)}(\kappa,t) \right|$ when $t = 0$ as a function of the wavenumber $k$ for $\kappa = 0.5$ (top left), $\kappa = 1$ (top right), $\kappa = \sqrt{2}$ (bottom left), and $\kappa = \sqrt{3}$ (bottom right) for $n = 0$ (black), $n = 1$ (blue), $n = 2$ (red), $n = 3$ (green), $n = 4$ (cyan), and $n = 5$ (magenta).}		\label{fou-akh-time-in-kappa}
\end{figure}
Figure~\ref{fou-akh-time-in-kappa} displays the moduli of the spatial Fourier amplitude spectrum of the Akhmediev breather when $t = 0$ as a function of wavenumber for several values of $\kappa$. Each panel shows only up to five pairs of sideband wavenumber as indicated by different colors, the first pair is blue, the second one is red, the third one is green, the fourth and fifth pairs are cyan and magenta, respectively. While the moduli of the sidebands are getting smaller as they progress toward higher-order pairs, the modulus corresponds to the main wavenumber does not always occupy the highest position. Depending on the sideband wavenumber $\kappa$, it might be taller than the first pairs of sidebands (for $\kappa_0 = \frac{1}{2}$ and $\kappa_0 = \sqrt{3}$) or shorter than the first pairs of sidebands (for $\kappa_0 = 1$ and $\kappa_0 = \sqrt{2}$). 

Different from the previous two cases where the spectra are continuous in $k$ (the KM breather and Peregrine soliton), the spectrum for this breather is discrete due to its periodic wave profile. While the initial modulated waves (when $t < 0$ but not $t \to -\infty$) exhibit the spectrum of the main/central frequency accompanied by a pair of sidebands, more pairs of sidebands were generated as the waves evolves in time. As time increases (when $t > 0$), higher-order sidebands in the spectrum disappear and return to the state of the initial condition when $t < 0$ with the central frequency with a pair of sidebands. This behavior suggests that the wave energy is distributed from the main wavenumber to the sideband wavenumbers and is recollected back during the evolution. Interestingly, for a particular value of $\kappa$, i.e., $\kappa_0 = 1$, all energy from the central wavenumber is transferred completely to its sidebands while for other values of $\kappa$, only partial energy is distributed from the central wavenumber to its sidebands.

\section{Conclusion}	\label{conclude}

In this article, we have considered the Fourier spectrum for rogue wave prototypes from the NLS equation. Also known as the soliton waves on a non-vanishing background, all three breathers are related with a complex-valued parameter. While the Peregrine soliton was discovered theoretically after the discovery of the KM breather and before the Akhmediev breather, it serves as the limiting case for both the KM and Akhmediev breathers. We need to perform integration in the complex plane when deriving the analytical expressions for the spatial Fourier spectrum of the breathers. Although the computations performed in this article are relatively straightforward, it turns out that the derivation is technically challenging nonetheless. 

Since both the KM and Peregrine breathers have infinite periodicity in the spatial domain, their spectra are continuous functions in both wavenumber and temporal domains. The spectrum modulus for the former in the temporal domain shows a pattern of periodicity while for the latter, the period is infinity. In the wavenumber domain, the corresponding spectrum moduli for both breathers exhibit infinite periodicity. For the Akhmediev breather, since it is periodic in space, its spatial spectrum remains continuous in the time domain but transforms into a discrete-type in the wavenumber domain. Remarkably, higher-order sidebands were generated as the initial sideband pairs evolve in time and then returned to an initial state. 

\subsection*{Dedication}
The author would like to dedicate this article to his late father Zakaria Karjanto (Khouw Kim Soey, 許金瑞) who introduced and taught him the alphabet, numbers, and the calendar in his early childhood. Karjanto senior was born in Tasikmalaya, West Java, Japanese-occupied Dutch~East~Indies on 1~January~1944 (Saturday~Pahing) and died in Bandung, West Java, Indonesia on 18~April~2021 (Sunday~Wage).

\subsection*{Acknowledgment}
The author gratefully acknowledges E. (Brenny) van Groesen for the long-lasting guidance and concerted cultivation in thinking and growing up mathematically.

\subsection*{Conflict of Interest}
The author declares no conflict of interest.



\begin{thebibliography}{99}
\bibitem{onorato2013rogue} Onorato, M., Residori, S., Bortolozzo, U., Montina, A., and Arecchi, F. T. (2013). Rogue waves and their generating mechanisms in different physical contexts. \textit{Physics Reports} 528(2): 47--89.	
	
\bibitem{dudley2014instabilities} Dudley, J. M., Dias, F., Erkintalo, M., and Genty, G. (2014). Instabilities, breathers and rogue waves in optics. \textit{Nature Photonics} 8(10): 755--764.	
	
\bibitem{sulem1999nonlinear} Sulem, C. and Sulem, P.-L. (1999). \textit{The Nonlinear Schr\"{o}dinger Equation--Self-Focusing and Wave Collapse}. Springer-Verlag: New York, NY, US.

\bibitem{fibich2015nonlinear} Fibich, G. (2015). \textit{The Nonlinear Schr\"odinger Equation--Singular Solutions and Optical Collapse}. Springer: Cham, Switzerland.

\bibitem{dysthe1999note} Dysthe, K. B., and Trulsen, K. (1999). Note on breather type solutions of the NLS as models for freak-waves. \textit{Physica Scripta} T82: 48--52.

\bibitem{kuznetsov1977solitons} Kuznetsov, E. A. (1977). Solitons in a parametrically unstable plasma. \textit{Doklady Akademii Nauk SSSR (Proceedings of the USSR Academy of Sciences)}, 236: 575--577.

\bibitem{ma1979perturbed} Ma, Y.-C. (1979). The perturbed plane-wave solutions of the cubic Schr\"{o}dinger equation. \textit{Studies in Applied Mathematics} 60(1): 43--58.

\bibitem{peregrine1983water} Peregrine, D. H. (1983). Water waves, nonlinear Schr\"{o}dinger equations and their solutions. \textit{The ANZIAM Journal} 25(1): 16--43.

\bibitem{akhmediev1985generation} Akhmediev, N., Eleonskii, V. M., and Kulagin, N. E. (1985). Generation of periodic trains of picosecond pulses in an optical fiber: exact solutions. \textit{Soviet Physics JETP} 62(5): 894-899.

\bibitem{akhmediev1986modulation} Akhmediev, N. N. and Korneev, V. I. (1986). Modulation instability and periodic solutions of the nonlinear Schrödinger equation. \textit{Theoretical and Mathematical Physics} 69(2): 1089--1093.

\bibitem{akhmediev1987exact} Akhmediev, N. N., Eleonskii, V. M., and Kulagin, N. E. (1987). Exact first-order solutions of the nonlinear Schr\"{o}dinger equation. \textit{Theoretical and Mathematical Physics} 72(2): 809--818.

\bibitem{karjanto2020peregrine} Karjanto, N. (2020). Peregrine soliton as a limiting behavior of the Kuznetsov-Ma and Akhmediev breathers. Accessible online at \url{https://arxiv.org/abs/2009.00269}, preprint arXiv:2009.00269 [nlin.PS]. Last accessed \today.

\bibitem{kibler2010peregrine} Kibler B., Fatome J., Finot C., Millot G., Dias F., Genty G., et al. (2010). The Peregrine soliton in nonlinear fibre optics. \textit{Nature Physics} 6: 790--795.

\bibitem{chabchoub2011rogue} Chabchoub A., Hoffmann N. P., Akhmediev N. (2011). Rogue wave observation in a water wave tank. \textit{Physical Review Letters} 106: 204502.

\bibitem{bailung2011observation} Bailung H, Sharma S. K., Nakamura Y. (2011). Observation of Peregrine solitons in a multicomponent plasma with negative ions. \textit{Physical Review Letters} 107: 255005.

\bibitem{alejo2020review} Alejo, M. A., Fanelli, L., and Muñoz, C. (2020). Review on the stability of the Peregrine and related breathers. \textit{Frontiers in Physics} 8: 591995.

\bibitem{pelinovsky2005spectral} Pelinovsky, D. (2005). Spectral analysis. In Scott, A. (Ed.), \textit{Encyclopedia of Nonlinear Science}, pp.~863--864. Routledge: New York, US and London, UK.

\bibitem{bauck2019note} Bauck, J. (2019). A note on Fourier transform conventions used in wave analyses. Available online at \url{https://engrxiv.org/jyt96/} and \url{https://doi.org/10.31224/osf.io/jyt96}. Last accessed \today. 
	
\bibitem{gradshteyn2015table} Gradshteyn, I. S. and Ryzhik, I. M. (2014). \textit{Table of Integrals, Series, and Products}, Eight Edition. Academic Press: Waltham, MA, US. Translated from Russian by Scripta Technica, Inc., D. Zwillinger (Editor) and V. Moll (Scientific Editor).
	
\bibitem{howie2003complex} Howie, J. M. (2003). \textit{Complex Analysis}. Springer-Verlag: London, UK, Berlin Heildelberg, Germany.

\bibitem{akhmediev1997solitons} Akhmediev, N. N., and Ankiewicz, A. (1997). \textit{Solitons: nonlinear pulses and beams}. Chapman \& Hall: London, UK.

\bibitem{karjanto2006thesis} Karjanto, N. (2006). \textit{Mathematical Aspects of Extreme Water Waves}. PhD thesis, University of Twente, the Netherlands. Accessible online at \url{https://arxiv.org/abs/2006.00766}, preprint arXiv:2006.00766 [nlin.PS]. Last accessed \today.
\end{thebibliography}
\end{document}